\documentclass{article}
\usepackage[margin=1in]{geometry}
\usepackage[T1]{fontenc}

\usepackage{amsmath,amsthm,amssymb}
\usepackage{mathrsfs}
\usepackage{xspace}
\usepackage{url}
\usepackage{hyperref}
\usepackage{appendix}
\usepackage[capitalize,noabbrev]{cleveref}
\usepackage{todonotes}
\usepackage{multirow}
\usepackage{tablefootnote}
\usepackage{afterpage}
\usepackage{lipsum}
\usepackage{booktabs}
\usepackage{graphicx}
\usepackage{subcaption}
\usepackage{paralist}

\hypersetup{colorlinks=true,allcolors=blue}

\theoremstyle{plain}
\newtheorem{theorem}{Theorem}[section]
\newtheorem{lemma}[theorem]{Lemma}

\theoremstyle{definition}
\newtheorem{definition}[theorem]{Definition}

\theoremstyle{remark}

\newtheorem*{remark*}{Remark}

\newcommand{\ignore}[1]{}

\DeclareMathOperator{\E}{\mathbb{E}}

\DeclareMathOperator{\R}{\mathbb{R}}
\DeclareMathOperator{\C}{\mathcal{C}}

\DeclareMathOperator{\poly}{poly}

\DeclareMathOperator{\cost}{cost}

\DeclareMathOperator{\OPT}{OPT}
\DeclareMathOperator{\ddim}{ddim}
\DeclareMathOperator{\tw}{tw}

\DeclareMathOperator{\dist}{dist}

\newcommand{\err}{\mathrm{err}}

\newcommand{\eat}[1]{\iffalse {#1}\fi}

\newcommand{\eps}{\varepsilon}
\renewcommand{\epsilon}{\varepsilon}

\newcommand{\calS}{\mathcal{S}}
\newcommand{\calX}{\mathcal{X}}

\newcommand{\ProblemName}[1]{\textsc{#1}}
\newcommand{\twoMedian}{\ProblemName{$2$-Median}\xspace}
\newcommand{\kMedian}{\ProblemName{$k$-Median}\xspace}

\newcommand{\kBMedian}{\ProblemName{$(k,\beta)$-Median}\xspace}
\newcommand{\kBHMedian}{\ProblemName{$(k,\beta/2)$-Median}\xspace}
\newcommand{\keMedian}{\ProblemName{$(k,\eta)$-Median}\xspace}

\newcommand{\Capx}{C^{\mathrm{apx}}}

\title{The Power of Uniform Sampling for $k$-Median}

\author{Lingxiao Huang\thanks{
    Email: \texttt{huanglingxiao1990@126.com}
  }\\
    Nanjing University
    \and
    Shaofeng H.-C. Jiang\thanks{
    Email: \texttt{shaofeng.jiang@pku.edu.cn}
  }\\
  Peking University
  \and
  Jianing Lou\thanks{ 
    Email: \texttt{loujn@pku.edu.cn}
  }\\
  Peking University
}

\begin{document}
\maketitle
\begin{abstract}
    We study the power of uniform sampling for \kMedian in various metric spaces.
    We relate the query complexity for approximating \kMedian,
    to a key parameter of the dataset, called the balancedness $\beta \in (0, 1]$ (with $1$ being perfectly balanced).
    We show that any algorithm must make $\Omega(1 / \beta)$ queries to the point set in order to achieve $O(1)$-approximation for \kMedian.
    This particularly implies existing constructions of coresets, a popular data reduction technique, cannot be query-efficient.
    On the other hand, we show a simple uniform sample of $\poly(k \epsilon^{-1} \beta^{-1})$
    points suffices for $(1 + \epsilon)$-approximation for \kMedian for various metric spaces, which nearly matches the lower bound.
    We conduct experiments to verify that in many real datasets, the balancedness parameter is usually well bounded, and that the uniform sampling performs consistently well even for the case with moderately large balancedness, which justifies that uniform sampling is indeed a viable approach for solving \kMedian.
\end{abstract}

\section{Introduction}
\label{sec:intro}
We investigate the power of uniform sampling in data reduction for \kMedian, which is a fundamental machine learning problem that has wide applications.
Given a metric space $(\calX,\dist)$, \kMedian takes an $n$-point dataset $X\subseteq \calX$ and an integer parameter $k\ge 1$ as inputs, and the goal is to find a $k$-point center set $C\subseteq \calX$ that minimizes the objective
\[
    \cost(X, C) := \sum_{x\in X}\dist(x, C),
\]
where $\dist(x,C):=\min_{c\in C}\dist(x,c)$ is the distance to the closest center.

Data reduction is a powerful way for dealing with clustering problems,
and a popular method called coreset~\cite{DBLP:conf/stoc/Har-PeledM04}
has been extensively studied during the last decades.
Roughly, an $\epsilon$-coreset aims to find a tiny but accurate proxy of the data, so that an $\epsilon$-approximate center set $C$ can be found by running existing algorithms on top of it. 
Specifically, coresets for \kMedian in Euclidean $\mathbb{R}^d$ has been studied in a series of works~\cite{DBLP:conf/stoc/Har-PeledM04,DBLP:journals/dcg/Har-PeledK07,DBLP:journals/siamcomp/Chen09,DBLP:conf/stoc/FeldmanL11, DBLP:journals/siamcomp/FeldmanSS20, DBLP:conf/focs/SohlerW18, DBLP:conf/stoc/HuangV20, DBLP:conf/soda/BravermanJKW21, DBLP:conf/stoc/Cohen-AddadSS21, DBLP:conf/focs/BravermanCJKST022},
and coresets of size $\poly(k\epsilon^{-1})$ were obtained, which is independent of the dimension $d$ and the data size.
In addition to speeding up existing algorithms, coresets can also be used to derive algorithms in sublinear modesl such as streaming~\cite{DBLP:conf/stoc/Har-PeledM04}, distributed~\cite{DBLP:conf/nips/BalcanEL13} and dynamic algorithms~\cite{DBLP:conf/esa/HenzingerK20}. 

Despite the progress on the size bounds, 
an outstanding issue of coresets is that known coreset \emph{constructions} are usually not query-efficient, i.e., it needs to access $\Omega(n)$ data points (even in sublinear models such as streaming).
In fact, it is not hard to see that this limitation cannot be avoided,
and in the worst case, any algorithm must make $\Omega(n)$ queries to the identity of data points in order to construct a coreset (see \Cref{thm:intro_lb} which we discuss later).
Technically, existing coreset constructions are often based on non-uniform sampling which is not data-oblivious, and this inherently requires to read the entire dataset.
This also makes it more difficult to efficiently implement in practice due to the sophisticated sampling procedure.
Hence, in order to achieve a sublinear query complexity, one must consider  other methods than the coreset.

To this end, we consider uniform sampling as a natural alternative data reduction approach for clustering.
Clearly, uniform sampling is data-oblivious and hence has a great potential to achieve sublinear query complexity.
Moreover, it often yields near-optimal solutions with only a few samples in practice, as was demonstrated 
by various experiments on real datasets in recent works on coresets for variants of clustering where uniform sampling is considered as a baseline (e.g.,~\cite{DBLP:conf/nips/MaromF19,DBLP:conf/icml/JubranTMF20,DBLP:conf/icml/BakerBHJK020,DBLP:conf/nips/BravermanJKW21, HJLW23}),
even though it is known that uniform sampling cannot yield a coreset in the worst case.

Thus, the focus of this paper is to understand the sampling complexity of uniform sampling for \kMedian and to justify its performance in practice.

\subsection{Our Results}
We first give a hardness result (\Cref{thm:intro_lb}) that even for $k = 2$ and 1D line, any algorithm, not only the uniform sampling, must make $\Omega(n)$ queries to the identity of points (which is the coordinate in 1D), in order to be $O(1)$-approximate to \kMedian.
In addition, \Cref{thm:intro_lb} further states that the number of queries
must depend on a parameter $\beta \in (0, 1]$, called \emph{balancedness} (\Cref{def:balance}), which is a property of the dataset.
Intuitively, $\beta$ measures how balance the optimal solution is, and precisely,
it requires the size of the smallest cluster in an optimal solution is at least $\beta$ times of the average cluster size $\frac{|X|}{k}$.

\begin{definition}[Balancedness]
    \label{def:balance}
    Given a dataset $X \subseteq \calX$, the balancedness $\beta \in (0, 1]$ of $X$ is the smallest number such that there is an optimal solution of \kMedian on $X$ satisfying that every cluster\footnote{For a center set $C = \{c_i\}_{i=1}^k$, each $c_i$ defines a cluster $X_i \subseteq X$ that consists of points whose nearest neighbor in $C$ is $c_i$.} has at least $\frac{\beta|X|}{k}$ points.
\end{definition}
The same notion of balancedness was considered in~\cite{DBLP:journals/ml/MeyersonOP04} which also studied the complexity of uniform sampling for \kMedian but achieved a weaker bound (which we will discuss later).
This balancedness was also enforced as a constraint to clustering problems~\cite{bradley2000constrained}, and more generally in capacitated clustering~\cite{DBLP:journals/jcss/CharikarGTS02} and fair clustering~\cite{DBLP:conf/nips/Chierichetti0LV17}.

\begin{theorem}
    \label{thm:intro_lb}
Any $O(1)$-approximate algorithm for \twoMedian with success probability at least $3/4$ must query the identify of data points in $X$ for $\Omega(1/\beta)$ times, where $\beta$ is the balancedness of the dataset $X$,
    even if the queried points have free access to the distance function.
\end{theorem}

Intuitively, the hard instance in \Cref{thm:intro_lb} has
two clusters, one has only one point and the other has many points with a small diameter, and the two clusters are very far away.
Then the uniform sampling fails to pick any point from singleton cluster with high probability, which incurs a big error. The detailed proof can be found in \Cref{sec:proof_intro_lb}.

\Cref{thm:intro_lb} suggests that the balancedness $\beta$ of the dataset may be a fundamental parameter that determines the necessary size of uniform sampling.
In our main result, stated in \Cref{thm:intro_ub}, we give a nearly-matching upper bound (with respect to $\beta$) for \kMedian in Euclidean $\mathbb{R}^d$,
which helps to justify this parameter is indeed fundamental.
This bound breaks the $\Omega(n)$ query complexity barrier of coresets,
and it readily yields sublinear-time algorithms for \kMedian.
The theorem may also be interpreted as beyond worst case analysis for uniform sampling, where the parameter $\beta$ provides a refined description to the structure of the dataset.
\begin{theorem}[Informal version of \Cref{thm:dim_ind}]
    \label{thm:intro_ub}
    Given a dataset $X\subset \R^d$ with balancedness $\beta \in (0, 1]$, $\epsilon \in (0, 0.5)$ and integer $k\ge 1$, let $S$ be $\tilde O(\frac{k^2}{\beta\epsilon^3})$  uniform samples\footnote{Throughout, $\tilde O(f) = O(f \poly \log f)$.} from $X$,
    then with probability $0.9$, one can find a $(1+\epsilon)$-approximation for \kMedian on $X$ only using $S$.
\end{theorem}

As mentioned, the dependence of $\beta$ in \Cref{thm:intro_ub} is nearly optimal.
Furthermore, the dependence of $k$ and $\epsilon^{-1}$ is only low-degree polynomial, and it also matches the known coreset size bounds.
Another feature of our bound is that it does not have a dependence in the Euclidean dimension $d$, thus it is very useful for dealing with high dimensional and/or sparse datasets.

In addition to the Euclidean case, we show a more general version (\Cref{thm:dim_ind}) that relates the sampling complexity to a notion of covering number (\Cref{def:covering}) which measures the complexity of the underlying metric space $(\calX, \dist)$.
By bounding the mentioned covering number (see \Cref{sec:application}), we also obtain similar bounds of $\poly(k \epsilon^{-1} \beta^{-1})$ for various other metric spaces such as doubling metrics and shortest-path metrics of graphs with bounded treewidth.
For general finite metric spaces, we obtain a bound of $\poly(k \epsilon^{-1}\beta^{-1} \log |\calX|)$.

Compared with the notion of coreset~\cite{DBLP:conf/stoc/Har-PeledM04}, the uniform sample $S$ may be interpreted as a coreset in a \emph{weaker} sense.
Specifically, coresets usually require the clustering cost be preserved for \emph{all} center sets $C \subseteq \calX$, while in \Cref{thm:intro_ub} we only guarantee the cost on near-optimal solutions, which still suffices for finding good approximation efficiently by running existing algorithms on $S$.
Another important difference is that for coresets, the size does not need to have a dependence in $\beta$ which we have, but as mentioned earlier, this size bound of coreset cannot be realized by query-efficient algorithms, which the uniform sampling does.

Indeed, similar notions of ``weak coreset'' was previously studied in the literature,
but the focus was mostly on the Euclidean case, and our bounds for doubling metrics and graph metrics are new.
Even for Euclidean spaces, only the special case of $k = 1$ was studied,
and previous bounds either depend on $d$~\cite{DBLP:journals/ki/MunteanuS18}
or have a worse $\epsilon^{-4}$ dependence than our $\epsilon^{-3}$~\cite{DBLP:conf/nips/Cohen-AddadSS21,danos21} (noting that $\beta = 1$ if $k = 1$).
For general metrics, \cite{DBLP:journals/ml/MeyersonOP04} gave a very similar size bound as in ours (which also depends on $\beta^{-1}$), but it only achieves $O(1)$ error instead of our $\epsilon$.
Somewhat less related, \cite{DBLP:conf/soda/MishraOP01,DBLP:journals/ml/Ben-David07,DBLP:journals/rsa/CzumajS07} gave uniform sampling bounds for the additive error guarantee, which is incomparable to ours.

Finally, even though \Cref{thm:intro_ub} implies a small uniform sample $S$ suffices for a sublinear-time algorithm for \kMedian,
to actually find the near-optimal solution on the sample $S$ can be tricky.
A natural idea is to find the optimal \kMedian on $S$ as the approximate solution for $X$,
but we show that this does not work, even when the dataset is balanced.
In particular, if one uses the optimal \kMedian on $S$ as the approximate solution, then it still requires $\Omega(n)$ samples even for a balanced dataset (see \Cref{lemma:kbmedian_lb}).

This motivates us to consider a variant of \kMedian, called \kBMedian, which aims to find the optimal center set $C$ subject to the constraint that $C$ is $\beta$-balanced (i.e., every cluster induced by $C$ has size at least $\beta \frac{|X|}{k}$; see \Cref{def:balanced_center}).
The balancedness constraint in \kBMedian is intuitive, since if the dataset is already balanced, then by definition there must be a balanced optimal solution, which \kBMedian can find.
We show in \Cref{thm:dim_ind} (i.e., the full statement of \Cref{thm:intro_ub}) that an $\alpha$-approximate \kBMedian on $S$ is $O(\alpha(1 + \epsilon))$-approximate \kMedian on $X$ with constant probability.

\paragraph{Experiments}
Our experiments focus on validating the performance of running an algorithm for \kMedian (instead of \kBMedian) on top of the uniform $S$,
since it is arguably the most natural approach and is likely to be used in practice.
Our experiments were conducted on various real datasets of different types of metric spaces including Euclidean $\mathbb{R}^d$ and shortest-path metrics.
We find that the solution returned by the \kMedian algorithm is fairly balanced, which effectively enforces the balancedness constraint of \kBMedian.
Moreover, we find that these datasets are mostly balanced, especially when the number of clusters $k$ is small.
Even when $k$ is relatively large and the balancedness becomes worse,
the factor of $1 / \beta$ in our upper bound is still reasonably bounded (\textasciitilde $100$),
and we also observe that the practical performance is not significantly worse than that of the small-$k$ case, if at all.
All these findings, together with our \Cref{thm:intro_ub}, justifies the effectiveness of uniform sampling in real datasets.

\subsection{Technical Overview}
Our proof of \Cref{thm:intro_ub} builds on a structural lemma (\Cref{lem:badsol}), which states that if a center set $C$ is ``bad'', i.e., its cost is larger than $(1 + \epsilon) \OPT$,
then this badness carries on to the sample.
Specifically, let $C^\star$ be the optimal center set,
if a center set $C$ satisfies $\cost(X, C) \geq (1 + \epsilon) \cost(X, C^\star)$, 
then we have $\cost(S, C) - \cost(S, C^\star) \geq \frac{O(\epsilon) |S|}{n} \cost(X, C^\star)$ with probability $1 - \exp(-O(|S|))$ (the big O hides dependence in other parameters such as $k$ and $\epsilon$).
Intuitively, conditioning on this event, one can conclude that any ``good'' center found in $S$ is also good in $X$, which implies the main theorem.

The special case $k = 1$ of the structural lemma was proved in~\cite{DBLP:journals/ki/MunteanuS18}, but the analysis does not seem to generalize our case $k \geq 2$.
Specifically, to bound $\cost(S, C) - \cost(S, C^\star)$,
it suffices to bound $\dist(x, C) - \dist(x, C^\star)$ for $x \in S$,
and in~\cite{DBLP:journals/ki/MunteanuS18} they can use triangle inequality $\dist(x, C) - \dist(x, C^\star) \leq \dist(C, C^\star)$ since $|C| = |C^\star| = 1$,
and the remaining analysis is on $\dist(C, C^\star)$ which does not depend on the variable $x$.
However, when $k \geq 2$ such a triangle inequality no longer holds,
and this forces us to use an alternative bound 
$\dist(x, C) - \dist(x, C^\star) \leq \max_{c\in C}\dist(c,C^\star)+\max_{c^\star\in C^\star}\dist(c^\star,C)$.
To analyze this, we crucially use a new observation of balancedness: if $C$ is a balanced center set, then $C$ and $C^\star$ must be ``close'' which depends on $1 / \beta$ (see \Cref{def:good} and \Cref{lem:balance_is_good}). This implies that both $\max_{c\in C} \dist(c, C^\star)$ and $\max_{c^\star \in C^\star} \dist(c^\star, C)$ are small enough,
and this bound eventually allows us to apply a concentration inequality to bound $\cost(S, C) - \cost(S, C^\star)$.
Indeed, the use of balancedness is a fundamental difference to~\cite{DBLP:journals/ki/MunteanuS18}.

Once the structural lemma is established, a natural next step is to apply it with a union bound on all centers that are close to $C^*$.
While these centers can still be infinitely many, one can apply standard discretization techniques, such as $\rho$-nets in $\mathbb{R}^d$, to reduce the number of events in the union bound.

\paragraph{Removing Dependence on $d$}
However, for Euclidean $\mathbb{R}^d$, a naive application of the net argument
only leads to a size bound that depends on $d$.
To remove the dependence on $d$ for the Euclidean case, we need a better union bound.
To this end, we use an alternative interpretation of the structural lemma.
In particular, we change the ``variable'' in the lemma from the center set $C$ to a vector $v^C := (\dist(x, C) - \dist(x, C^\star))_{x \in X} \in \mathbb{R}^X$ which represents all the relevant costs that $C$ induce.
Then the structural lemma is equivalently stated as: if some $v \in \mathbb{R}^X$ satisfies $\|v\|_1 \geq \epsilon \cdot \cost(X, C^\star)$,
then $\|v|_{S}\|_1 \geq \frac{O(\epsilon) |S|}{n} \cost(X, C^\star)$ with high probability,
where $v|_{S}$ means restricting $v$ to the uniform sample $S$.
Now, we try to apply the union bound on a discretization of the cost vectors $\{ v^C \}_{C}$, instead of the space of all center points $C$.

Specifically, for a fixed sample $S$, we need to find a discretized set $U$ such that for every $v$, $U$ contains a vector $v'$ with $\|v|_S\|_1 \approx \|v'|_S\|_1$ (recalling that we do not need to preserve $\|v\|_1$).
Note that this $U$ needs not be a subset of $\{v^C\}_C$. In other words, a vector $u \in U$ may not correspond to a center set $C \subset \calX$,
and they can be any real vector in $\mathbb{R}^X$.
This turns out to be great freedom compared with discretization of center sets.
Indeed, for Euclidean $\mathbb{R}^d$, to preserve $\|v|_S\|$,
we can map $S \cup C^\star$  into a low dimensional space of dimension $d' := O(\log(|S \cup C^\star|) / \epsilon^2)$ using a terminal version of the Johnson-Lindenstrauss transform~\cite{DBLP:conf/stoc/NarayananN19},
and discretize only in this lower dimensional space.
This removes the dependence in $d$ and replaces it with $d'$.
In addition, we use a chaining argument~\cite{talagrand1996majorizing} 
which was also used in several recent works about coresets~\cite{DBLP:conf/nips/Cohen-AddadSS21,DBLP:conf/stoc/Cohen-AddadLSS22, DBLP:journals/corr/abs-2211-11923},
to further save an $\epsilon^{-1}$ factor and obtain $\epsilon^{-3}$ dependence in the final bound.
Compared with a closely related work~\cite{DBLP:conf/nips/Cohen-AddadSS21},
our chaining argument is applied on the entire $X$, while theirs is applied on $O(\epsilon^{-2})$ pieces of $X$ separately which results in an addition $\epsilon^{-1}$ factor more than ours.
Finally, we note that it may cause randomness issues since we assumed fixed sample $S$ before finding $U$. Luckily, we manage to fix this by relating it with a Gaussian process.

\paragraph{Beyond Euclidean Spaces}
In fact, the above argument is useful not only for removing the dependence on $d$ for Euclidean spaces,
but also for obtaining bounds for other metric spaces.
We show it suffices to bound the size of $U$ in order to obtain a bound for uniform sampling.
We formulate this minimum size of $U$ as the covering number (see \Cref{def:covering}),
and we derive covering number bounds for several types of metric spaces.
Previously, A similar notion of covering number as well as its use to bound the coreset size 
was also considered in the coreset literature~\cite{DBLP:conf/stoc/Cohen-AddadLSS22, DBLP:journals/corr/abs-2211-11923}.
We give a more detailed comparison in \Cref{sec:covering}.

\subsection{Related Work}

Stemming from~\cite{DBLP:conf/stoc/Har-PeledM04},
the study of size bounds for coresets has been very fruitful.
For \kMedian in Euclidean $\mathbb{R}^d$, a series of works improves the size from $O(\poly(k) \epsilon^{-d} \log n)$ all the way to $\poly(k\epsilon^{-1})$~\cite{DBLP:conf/stoc/Har-PeledM04,DBLP:journals/dcg/Har-PeledK07,DBLP:journals/siamcomp/Chen09,DBLP:conf/stoc/FeldmanL11,DBLP:journals/siamcomp/FeldmanSS20,DBLP:conf/focs/SohlerW18,DBLP:conf/nips/Cohen-AddadSS21,DBLP:conf/stoc/Cohen-AddadLSS22}.
Recent works focus on deriving tight bounds for $k$ and $\epsilon^{-1}$.
The state-of-the-art coresets for \kMedian in $\R^d$ achieves a size of $\tilde O(\min\{k \eps^{-3},k^{\frac{4}{3}}\epsilon^{-2}\})$~\cite{DBLP:conf/stoc/Cohen-AddadLSS22,DBLP:journals/corr/abs-2211-11923}, which nearly matches a lower bound of $\Omega(k\epsilon^{-2})$~\cite{DBLP:conf/stoc/HuangV20,DBLP:conf/stoc/Cohen-AddadLSS22}.
Beyond Euclidean spaces,  coresets for \kMedian were obtained in doubling metrics~\cite{DBLP:conf/focs/HuangJLW18}, shortest-path metrics of graphs with bounded treewidth~\cite{DBLP:conf/icml/BakerBHJK020} and graphs that exclude a fixed minor~\cite{DBLP:conf/soda/BravermanJKW21}.
In addition, coresets for variants of clustering have also been studied, notably capacitated clustering and the highly related fair clustering~\cite{DBLP:conf/waoa/0001SS19, DBLP:conf/nips/HuangJV19, DBLP:conf/icalp/BandyapadhyayFS21, DBLP:conf/focs/BravermanCJKST022}, and
robust clustering~\cite{DBLP:conf/soda/FeldmanS12,HJLW23}.

 \section{Preliminaries}
\label{sec:pre}

We define \kBMedian problem in \Cref{def:kBMedian} which depends on the notion of balanced center sets (\Cref{def:balanced_center}).
The notion of weak coreset (\Cref{def:weak_coreset}) captures the main guarantee of \Cref{thm:dim_ind}.

\begin{definition}[Balanced Center Set]
	\label{def:balanced_center}
	Given a dataset $X\subseteq \calX$, an integer $k\geq 1$ and $\beta\in (0,1)$, we say a center set $C\subseteq \calX$ is $\beta$-balanced if for every cluster $X_i$ ($i\in [k]$) induced by $C$ contains at least $\beta |X|/k$ points.
Let $\C_\beta(X)$ denote the collection of all $\beta$-balanced center sets on $X$. 
\end{definition}

Recall that if the balancedness of $X$ is $\beta$, there exists an optimal solution of \kMedian on $X$ that is $\beta$-balanced.
Then we naturally define the following \kBMedian problem that aims to find optimal solutions within $C_\beta(X)$.

\begin{definition}[\kBMedian]
	\label{def:kBMedian}
	Given a dataset $X\subseteq \calX$, an integer $k\geq 1$ and $\beta\in (0,1]$, the goal of the \kBMedian problem is to find a $k$-point set $C\in\C_\beta(X)$ that minimizes $\cost(X,C)$.
Let $\OPT_{\beta}(X):=\cost(X,C^\star)$ where $C^\star$ is an optimal solution for \kBMedian on $X$.
\end{definition}

Again, if the balancedness of $X$ is $\beta$, solving \kBMedian also solves \kMedian and $\OPT_{\beta}(X)$ equals the optimal \kMedian cost on $X$.
A similar notion of \kBMedian also appears in the literature, e.g.,~\cite{bradley2000constrained, DBLP:conf/sspr/MalinenF14, DBLP:journals/isci/CostaAM17, DBLP:conf/ijcai/LinHX19, DBLP:journals/tcs/Ding20}. 
The main difference is that they allow points being assigned to a non-closest center for achieving a balanced dataset partition, and hence, consider all possible $k$ points sets instead of $\C_{\beta}(X)$.

\begin{definition}[Weak Coreset for \kBMedian]
	\label{def:weak_coreset}
	Given a dataset $X\subseteq\calX$, an integer $k\geq 1$ and $\beta,\eps \in (0,1]$, an $\eps$-weak coreset for \kBMedian on $X$ is a subset $S\subseteq X$ such that for every $k$-point set $C\in\C_{\beta/2}(S)$ with
	$
	\sum_{x\in S}\dist(x, C)\le (1+\epsilon)\OPT_{\beta/2}(S),
	$
	it holds that
	$
	\sum_{x\in X}\dist(x,C)\le (1+O(\epsilon))\OPT_{\beta}(X).
	$
\end{definition}

Intuitively, the above definition requires that any near-optimal solution for \kBHMedian on a weak coreset $S$ is a near-optimal solution for \kBMedian on $X$.
Consequently, solving \kBHMedian on $S$ leads to a near-optimal solution for \kMedian on $X$ if the balancedness of $X$ is $\beta$.
The points are unweighted in our weak coreset, as opposed to the weighted points considered in (strong) coresets~\cite{DBLP:conf/stoc/Har-PeledM04}.
This is natural since we consider uniform sampling.

Note that an optimal solution $C^\star$ for \kBMedian on $X$ may not be $\beta$-balanced on $S$ due to the small size of $S$, and hence, we consider a relaxed balancedness $\beta/2$ instead of $\beta$ for $S$ such that the considered collection $\C_{\beta/2}(S)$ is likely to include $C^\star$.

 \section{Uniform Sampling Yields Weak Coreset}
\label{sec:algorithm}

\begin{theorem}[Main Theorem]
    \label{thm:dim_ind}
    Let $(\calX,\dist)$ be a metric space and $X\subseteq \calX$ be a dataset. 
Given an integer $k\ge 1$ and real numbers $\beta\in(0,1],\epsilon\in (0,0.5)$, let integer $m$ satisfy that
    \begin{equation}
        \label{eq:m_X}
        m\ge O\left(\frac{k}{\beta\eps^2}\left(\sum_{i=1}^{\log\eps^{-1}}\sqrt{2^{-i}\log N^{2^{-i}}_X(m)} \right)^2\right)
    \end{equation}
    where $N_X^{\alpha}(m)$ is the covering number defined in \Cref{def:covering}. Then, a set $S$ of $m$ uniform samples from $X$ is an $\epsilon$-weak coreset for \kBMedian on $X$ with probability at least $0.9$.
\end{theorem}

The factor $\sum_{i=1}^{\log\eps^{-1}}\sqrt{2^{-i}\log N^{2^{-i}}_X(m)}$ plays a similar role as the entropy integral (or Dudley integral) that are commonly used in
the chaining argument (see e.g., Corollary 5.25 of~\cite{Handel2014ProbabilityIH}).
A more concise sufficient condition for \eqref{eq:m_X} is
\begin{equation}
    \label{eq:simple_m_X}
    m\ge O\left(\frac{k}{\beta\eps^2}\cdot \log N_X^{\eps}(m) \right)
\end{equation}
which is derived directly by the monotonicity of the covering number, i.e., $N^{2^{-i}}_X(m)\le N^{\eps}_X(m)$ for all $i\le \log\eps^{-1}$. 
However, we still need to use ~\eqref{eq:m_X} in order to obtain better sample complexity, especially for Euclidean spaces. We give bounds for this term for various metrics in \Cref{sec:application}.

\paragraph{Proof Overview}
To utilize the balancedness, we consider \emph{good center sets} $\C'(X)$ (\Cref{def:good}) as a collection of center sets $C$ that are ``close'' to $C^\star$. 
We show that any near-optimal center set for $(k,\beta/2)$-Median on $S$ is likely to be good (\Cref{lem:balance_is_good}).
Then it suffices to show that all center sets $C\in \C'(X)$ with $\cost(X,C) \geq (1 + O(\eps))\OPT_\beta(X)$, denoted as $\C^\mathrm{bad}(X)$, are likely to have a large cost on $S$, i.e., $\cost(S,C) > (1+\eps) \OPT_{\beta/2}(S)$. 
In other words, we need a uniform convergence guarantee on $\C^\mathrm{bad}(X)$.
To this end, we need to bound the ``complexity'' of $\C^\mathrm{bad}(X)$,
by considering the notion of \emph{covering}, which may be viewed as a set of representatives, and the \emph{covering number} which measures the complexity of the covering (\Cref{def:covering}).
In the last steps of the proof (\Cref{sec:mainproof}), we reduce the above requirement for $\C^\mathrm{bad}(X)$ to a Gaussian process and applies a chaining argument based on the covering.

\subsection{Good Event and Good Center Sets}
\label{sec:candi}

We first introduce some useful notations.
Let $\lambda>1000$ be a constant throughout this section. 
For a center set $C\subseteq \calX$ and $x\in \calX$, denote by $C(x)=\arg\min_{c\in C}\dist(c,x)$ the closest center in $C$ to $x$ (breaking ties arbitrarily). 
Let $C^\star$ denote an optimal center set of $X$ for \kBMedian. 
For a subset $A\subseteq X$, and real numbers $\eta \in (0,1),\alpha>0$,  denote by 
$
    \C_{\eta}^{(\alpha)}(A):=\{C\in \C_{\eta}(A):\cost(A,C) \le (1+\alpha)\OPT_{\eta}(A)\}    
$
the set of all $(1+\alpha)$-approximate $k$-point set for \keMedian on $A$, and let $\overline{\C}_{\eta}^{(\alpha)}(A)=\C_{\eta}(A)\setminus \C_{\eta}^{(\alpha)}(A)$. 
By definition, we know that $S$ is an $\eps$-coreset weak coreset for \kBMedian on $X$ if and only if
\begin{equation}
    \label{eqn:iff}
    \C_{\beta/2}^{(\eps)}(S)\cap \overline{\C}_{\beta}^{(O(\eps))}(X) = \emptyset.
\end{equation}

Let $\mathcal{P}^\star=\{X_1^\star,...,X_k^\star\}$ be the partition of $X$ induced by $C^\star$. 
We denote by $\xi_S$ the event that
\begin{equation}
    \label{eq:xi_S}
    \begin{aligned}
    &\frac{1}{m}\cost(S,C^\star)\le \lambda\cdot\frac{1}{n}\OPT_\beta(X)\quad
    \wedge\quad\forall i\in[k], \frac{|S\cap X_i^\star|}{|S|}\in (1\pm\frac{1}{2})\frac{|X_i^\star|}{|X|},
    \end{aligned}
\end{equation}
where the first condition requires that the average \kMedian cost of $S$ to $C^\star$ is not too large compared to that of $X$, and the second condition requires that the ratio of sampled points in every cluster $X_i^\star$ is close to the underlying one.
The following lemma claims that $\xi_S$ happens with high probability, and hence, we can condition on $\xi_S$ in the analysis.

\begin{lemma}
    \label{lem:Pr_xi_S}
    $\xi_S$ happens with probability at least 0.99.
\end{lemma}
\begin{proof}
	Since $C^\star$ is a $\beta$-balanced center set on $X$, we have $|X_i^\star|\ge \frac{\beta n}{k}$ for every $i\in[k]$. Recall that $S$ is a set of uniform samples, therefore by Chernoff bound, we have 
	\begin{equation*}
	\Pr\left[\frac{|S\cap X_i^\star|}{|S|}\not \in \left(1\pm \frac{1}{2}\right)\frac{|X_i^\star|}{|X|}\right]
	\le 2\exp\left(-\frac{\beta m}{12k}  \right)
	\le 0.001/k.
	\end{equation*}
	By the union bound, we have 
	\begin{equation*}
	\Pr\left[\forall i\in [k], \frac{|S\cap X_i^\star|}{|S|}\in \left(1\pm \frac{1}{2}\right)\frac{|X_i^\star|}{|X|}\right]\ge 0.999.
	\end{equation*}
	Also note that $\E_S[\cost(S,C^\star)]=\frac{m}{n}\OPT_\beta(X)$. 
Then by the Markov inequality, with probability at least $0.999$, we have $\cost(S,C^\star)\le \lambda\cdot\frac{m}{n}\OPT_\beta(X)$ since $\lambda>1000$. 
This completes the proof.
\end{proof}

\begin{definition}[Good Center Sets]
    \label{def:good}
We say a $k$-point set $C\subseteq \calX$ is \emph{good} if we have
    \begin{equation}
        \label{eq:b1}
        \frac{1}{n}\sum_{x\in X}\dist(C^\star(x),C)\le \frac{6\lambda}{n}\OPT_\beta(X),
    \end{equation}
    and for every $x\in X$,
    \begin{equation}
        \label{eq:b2}
        \begin{aligned}
        \left|\dist(x,C)-\dist(x,C^\star) \right| \le \frac{6\lambda k}{\beta n} \OPT_\beta(X).
        \end{aligned}
    \end{equation}
    Let 
    $
        \mathcal{C}'(X):=\{C\subseteq \calX:|C|=k, C\text{ satisfies \eqref{eq:b1},\eqref{eq:b2}} \}
    $
    denote the collection of all good center sets on $X$.
\end{definition}

Intuitively, we say a center set $C$ is good if it is ``close'' to $C^\star$. 
\eqref{eq:b1} means that the average distance from every $C^\star(x)$ to $C$
is not too large, and~\eqref{eq:b2} states that all distance differences $|\dist(x, C)-\dist(x, C^\star)|$ are small. 
We note that the definition of good center sets is independent of $S$, which is useful for the probability analysis on $S$.

The following lemma states that all $(1+\epsilon)$-approximate center sets for \kBHMedian on $S$ are good conditioning on $\xi_S$. 
Recall that we want to prove \Cref{eqn:iff}, hence by \Cref{lem:balance_is_good}, it remains to prove $\C_{\beta/2}^{(\eps)}(S)\cap \left(\C'(X)\cap \overline{\C}_{\beta}^{(O(\eps))}(X)\right) = \emptyset$, which is easier to handle due to good properties of $\C'(X)$.

\begin{lemma}
    \label{lem:balance_is_good}
     $\mathcal{C}_{\beta/2}^{(\epsilon)}(S)\subseteq \C'(X)$ holds conditioning on $\xi_S$.
\end{lemma}

\begin{proof}
	Conditioning on $\xi_S$, we have $C^\star\in \C_{\beta/2}(S)$ and for every $C\in\C_{\beta/2}^{(\epsilon)}(S)$, it holds that
	\begin{equation}
	\label{eq:sumCstar}
	\sum_{x\in S}\dist(C^\star(x),C)\le  \sum_{x\in S}\left(\dist(x,C^\star)+\dist(x,C)\right)
	\le  (2+\epsilon)\sum_{x\in S}\dist(x,C^\star)
	\le  \frac{3\lambda m}{n}\OPT_\beta(X),
	\end{equation}
	where the first derivation is due to the triangle inequality, the second derivation is because $\sum_{x\in S}\dist(x,C)\le (1+\eps)\OPT_{\beta/2}(S)\le (1+\eps)\sum_{x\in S}\dist(x,C^\star)$, and the last derivation is due to \eqref{eq:xi_S}.
	
	We can rewrite $\sum_{x\in S}\dist(C^\star(x),C)$ as $\sum_{c^\star_i\in C^\star}\dist(c^\star_i,C)\cdot |S\cap X_i^\star|$, where $c^\star_i$ denotes the center of cluster $X_i^\star$. Therefore, we have by \eqref{eq:xi_S},
	\begin{equation*}
	\sum_{x\in S}\dist(C^\star(X),C)\ge \frac{m}{2n}\sum_{c^\star_i\in C^\star}\dist(c^\star_i,C)\cdot |X_i^\star|
	= \frac{m}{2n}\sum_{x\in X}\dist(C^\star(x),C),
	\end{equation*}
	which combining with \eqref{eq:sumCstar} completes the proof of \eqref{eq:b1}.
	
	To prove \eqref{eq:b2}, we observe that for every $x\in X$, it holds that
	\begin{equation*}
	\left|\dist(x,C)-\dist(x,C^\star) \right|
	\le \max\{\dist(C(x),C^\star),\dist(C^\star(x),C)\}. 
	\end{equation*}
	We only upper bound $\dist(c,C^\star)$ for every $c\in C$, and bounding $\dist(c^\star,C)$ for every $c^\star\in C^\star$ is almost the same. Since $C\in\C_{\beta/2}(S)$, for every $c\in C$, we have
	\begin{eqnarray*}
		\dist(c,C^\star)&\le& \sum_{c\in C}\dist(c,C^\star)\\
		&\le&\frac{2k}{\beta}\cdot\frac{1}{m}\sum_{c\in C}\frac{\beta m}{2k}\dist(c,C^\star)\\
		&\le&\frac{2k}{\beta}\cdot \frac{1}{m}\sum_{x\in S}\dist(C(x),C^\star)\\
		&\le&\frac{2k}{\beta}\cdot \frac{1}{m}\sum_{x\in S}\left(\dist(x,C)+\dist(x,C^\star) \right)\\
		&\le&\frac{2k}{\beta}\cdot \frac{2+\epsilon}{m}\sum_{x\in S}\dist(x,C^\star)\\
		&\le&\frac{6\lambda k}{\beta n}\OPT_\beta(X),
	\end{eqnarray*}
	where the third derivation is because $\sum_{x\in S}\dist(C(x), C^\star)=\sum_{c_i\in C}\dist(c_i,C^\star)\cdot |X_i|$ for $\{X_1,\dots,X_k\}$ being the partition of $X$ induced by $C$, and $|X_i|\ge \frac{\beta m}{2k}$ since $C\in \C_{\beta/2}(S)$.
\end{proof}

\subsection{Covering and Covering Number}
\label{sec:covering}
The notion of covering and covering number, defined in  \Cref{def:covering}, plays a crucial role in our analysis. 
We start with giving the definition, and then discuss how the several relevant parameters are chosen for our application.

\begin{definition}[Covering and Covering Number]
    \label{def:covering}
    Given a dataset $X\subseteq \calX$ and a subset $S\subseteq X$, a set of vectors $V\subset \R^X$, an error function $\err:X\times V\to \R$ and real numbers $\alpha,\gamma>0$, we say $U\subset \R^X$ is a $\gamma$-bounded $\alpha$-covering of $V$ w.r.t. $(S,\err)$ if the following holds:
    \begin{compactenum}
        \item  (Bounded Covering Error) for every $v\in V$, there exists a vector $u\in U$ such that
        \begin{equation*}
            \forall x\in S,\quad \left|v_x-u_x \right|\le \alpha\cdot\err(x,v)
        \end{equation*}
        \item  (Bounded $L_\infty$ Norm) for every $u\in U$, $\|u\|_\infty\le \gamma$.
    \end{compactenum}
    Define $N^{\alpha,\gamma}(S,V,\err)$ to be the minimum cardinality $|U|$ of any $\gamma$-bounded $\alpha$-covering $U$ of $V$ w.r.t. $(S,\err)$. 
Moreover, let $\calS\subseteq 2^X$ be a collection of subsets and define the $\gamma$-bounded $\alpha$-covering number of $V$ w.r.t. $(\calS,\err)$ to be 
    \begin{equation*}
        N_X^{\alpha,\gamma}(\mathcal{S},V,\err):=\max_{S\in \mathcal{S}} N^{\alpha,\gamma}(S,V,\err)
    \end{equation*}
\end{definition}

\paragraph{Explanation of \Cref{def:covering}}
    The idea of $\epsilon$-covering has also been used in the coreset literature, e.g.,~\cite{DBLP:conf/stoc/Cohen-AddadLSS22,DBLP:journals/corr/abs-2211-08184,DBLP:journals/corr/abs-2211-11923}.
Intuitively, the covering may be viewed as a discretization/representative of $V$,
    and the covering number measures its complexity.
The parameter $\alpha$ together with the error function $\err$ controls the granularity of the discretization of $V$,
    and the covering number $N^{\alpha,\gamma}(S,V,\err)$ increases as $\alpha$ becomes larger. 
The relative errors $\err(x,v)$ should often be tailored to the application (e.g.,~\cite{DBLP:conf/stoc/Cohen-AddadLSS22,DBLP:journals/corr/abs-2211-08184,DBLP:journals/corr/abs-2211-11923}) and we need to use a specific definition of it.
Compared with a standard definition of covering, we additionally require $\|u\|_\infty$ for all $u\in U$ bounded (by parameter $\gamma$). 
This requirement is useful for bounding the variance of a Gaussian process in our analysis, which plays a similar role as excluding huge subsets as in the definition of covering in \cite{DBLP:journals/corr/abs-2211-11923} (their Definition 3.2).
A natural choice of $\mathcal{S}$ is the collection of all $S\subseteq X$ with a fixed cardinality
but in our case we need additional constraints on $\mathcal{S}$ to bound the overall covering errors.

\paragraph{Specifying $V$, $\gamma$, $\err$ and $\mathcal{S}$} 
For a center set $C\subseteq \calX$, define $v^C\in \R^X$ to be a cost vector such that 
\[
    \forall x\in X,\quad v^C_x=\dist(x,C)-\dist(x,C^\star),
\]
and this is motivated by~\eqref{eq:b2} which considers the difference of the distances. 
Since our goal is to prove $\C'(X)\subseteq \C_{\beta}^{(O(\eps))}(X)$, we consider the following $V$ on good center sets:
\[
    V=\{v^C:C\in\C'(X)\}.
\] 
As~\eqref{eq:b2} implies $\|v^C\|_{\infty}\leq \frac{6\lambda k}{\beta n}\OPT_{\beta}(X)$, we select 
\[
\gamma = \frac{12\lambda k}{\beta n}\OPT_{\beta}(X).
\]
Now we define function $\err: X\times V\rightarrow \R$.
For every $x\in X$ and $v^C\in V$, 
\begin{align*}
    &\quad\err(x,v^C)\\
    =&\quad v^C_x + 2\dist(x,C^\star) + \frac{1}{n}\OPT_\beta(X)\\
    =&\quad \dist(x,C)+\dist(x,C^\star) + \frac{1}{n}\OPT_\beta(X).
\end{align*}

\noindent
The term $\frac{1}{n}\OPT_\beta(X)$ is consistent with the selection of $\gamma$. 
The term $\dist(x,C)+\dist(x,C^\star)$ is mainly designed for obtaining a dimension-independent covering number in Euclidean spaces; see \Cref{lem:covering_euclidean} for details.

Finally, we specify $\calS$ with an additional restriction $\xi_S$. 
\[
    \mathcal{S}(m) = \{S\subseteq X: |S|\le m, \xi_S \}.
\] 
We shorten the notation of the covering number by
\[
N_X^{\alpha}(m):=N_X^{\alpha,\gamma}\left(\calS(m),V,\err\right),
\]
and call it the $\alpha$-covering number.
We have the following lemma showing that $\xi_S$ leads to a bound of the total covering error, which is helpful for bounding the variance of our Gaussian process.

\begin{lemma}[$\xi_S$ Implies Bounded Covering Error]
    \label{lem:bounderror}
    For $S$ such that $\xi_S$ happens, we have for every $C\in \C'(X)$,
    \begin{equation*}
        \sum_{x\in S}\err(x, v^C)\le \frac{15\lambda m}{n}\OPT_\beta(X).
    \end{equation*}
\end{lemma}

\begin{proof}
	$\xi_S$ implies that $\cost(S,C^\star)\le \lambda\cdot\frac{m}{n}\OPT_\beta(X)$ and 
	\begin{align*}
	\sum_{x\in S}\dist(C^\star(x),C)&=\sum_{c^\star_i\in C^\star} \dist(c^\star_i,C)\cdot|S\cap X_i^\star|\\
	&\le \sum_{c^\star_i\in C^\star}\dist(c^\star_i,C)\cdot 2|X_i^\star|\cdot\frac{m}{n}\\
	&\le \frac{2m}{n}\sum_{x\in X}\dist(C^\star(x),C)\\
	&\le \frac{12\lambda m}{n}\OPT_\beta(X).
	\end{align*}
	where the second derivation is due to~\eqref{eq:xi_S}, and the forth derivation is due to~\eqref{eq:b1}.
	Therefore for every $C\in \C'(X)$, it holds that
	\begin{align*}
	\sum_{x\in S}\err(x,v^C)
	&=\sum_{x\in S}\left(\dist(x,C)+\dist(x,C^\star) + \frac{1}{n}\OPT_\beta(X) \right)\\
	&\le \frac{m}{n}\OPT_\beta(X)+\sum_{x\in S}\left(2\dist(x,C^\star)+\dist(C^\star(x),C) \right)\\
	&\le \frac{15\lambda m}{n}\OPT_\beta(X)
	\end{align*}
\end{proof}

\subsection{Proof of Main Theorem: \Cref{thm:dim_ind}}
\label{sec:mainproof}
Conditioning on $\xi_S$, for every $C\in \C_{\beta/2}^{(\epsilon)}(S)$, it holds that
\begin{equation}
    \label{eq:goodsol}
    \sum_{x\in S}\dist(x,C)\le (1+\epsilon)\sum_{x\in S}\dist(x,C^\star)
    \le \sum_{x\in S}\dist(x,C^\star)+\frac{\lambda\epsilon m}{n}\OPT_\beta(X),
\end{equation}
where the second inequality is due to \eqref{eq:xi_S}.
Let
\[
    \C^\mathrm{bad}(X):=\C'(X)\cap \overline{\C}_\beta^{(10\lambda^2\epsilon)}(X)
\]
denote the collection of all good solutions $C$ which are bad on $X$, i.e., $\cost(X,C)\ge (1+10\lambda^2\eps)\OPT_\beta(X)$. To show $S$ is an $\epsilon$-weak coreset, let $\phi_S$ denotes the event that for every $C\in\C^\mathrm{bad}(X)$,~\eqref{eq:goodsol} is far from being satisfied, that is,
    $\sum_{x\in S}\dist(x,C)\ge \sum_{x\in S}\dist(x,C^\star)+\frac{\lambda^2\epsilon m}{n}\OPT_\beta(X)$.
Then it suffices to bound the following probability
\begin{align*}
    \Pr_S[\C^{(\epsilon)}_{\beta/2}(S)\cap \overline{\C}_\beta^{(10\lambda^2\epsilon)}(X)\neq \emptyset ] &= \Pr_S[\C^{(\epsilon)}_{\beta/2}(S)\cap \overline{\C}_\beta^{(10\lambda^2\epsilon)}(X)\neq \emptyset\wedge \phi_S ]\\
    &\quad+\Pr_S[\C^{(\epsilon)}_{\beta/2}(S)\cap \overline{\C}_\beta^{(10\lambda^2\epsilon)}(X)\neq \emptyset\wedge \neg\phi_S]\\
    &\le \Pr_S[\C^{(\epsilon)}_{\beta/2}(S)\cap \overline{\C}_\beta^{(10\lambda^2\epsilon)}(X)\neq \emptyset\wedge \phi_S ]+ \Pr_S[\neg\phi_S].
\end{align*}
For the first term, we have 
\begin{align*}
    \Pr_S[\C^{(\epsilon)}_{\beta/2}(S)\cap \overline{\C}_\beta^{(10\lambda^2\epsilon)}(X)\neq \emptyset\wedge \phi_S ]
    &= \Pr_S[\C^{(\epsilon)}_{\beta/2}(S)\cap \overline{\C}_\beta^{(10\lambda^2\epsilon)}(X)\neq \emptyset\wedge \phi_S \mid \xi_S ]\Pr[\xi_S]\\
    &\quad+\Pr_S[\C^{(\epsilon)}_{\beta/2}(S)\cap \overline{\C}_\beta^{(10\lambda^2\epsilon)}(X)\neq \emptyset\wedge \phi_S \mid \neg\xi_S ]\Pr_S[\neg \xi_S]\\
    &\le \Pr_S[\C^{(\epsilon)}_{\beta/2}(S)\cap \overline{\C}_\beta^{(10\lambda^2\epsilon)}(X)\neq \emptyset\wedge \phi_S \mid \xi_S ]+\Pr_S[\neg \xi_S]
\end{align*}
By the discussion above, conditioning on $\xi_S$, it holds that $\C_{\beta/2}^{(\epsilon)}(S)\subseteq \C'(X)$. So $\C^{(\epsilon)}_{\beta/2}(S)\cap \bar{\C}_0^{(10\lambda^2\epsilon)}(X)\neq \emptyset$ is equivalent to $\C^{(\epsilon)}_{\beta/2}(S)\cap \C^\mathrm{bad}(X)\neq \emptyset$, which implies that there exists $C\in \C^\mathrm{bad}(X)$ such that~\eqref{eq:goodsol} holds. We know this contradicts $\phi_S$. Therefore, we have 
\[
    \Pr_S[\C^{(\epsilon)}_{\beta/2}(S)\cap \overline{\C}_\beta^{(10\lambda^2\epsilon)}(X)\neq \emptyset\wedge \phi_S \mid \xi_S ]=0,
\]
and thus
\[
    \Pr_S[\C^{(\epsilon)}_{\beta/2}(S)\cap \overline{\C}_\beta^{(10\lambda^2\epsilon)}(X)\neq \emptyset\wedge \phi_S ]\le \Pr_S[\neg \xi_S]\le 0.01,
\]
where the second inequality follows from \Cref{lem:Pr_xi_S}.

It remains to bound $\Pr_S[\neg\phi_S]\le 0.09$. Recall that for any $C\in \C^\mathrm{bad}(X)$, $\|v^C\|_1=\sum_{x\in X}(\dist(x,C)-\dist(x,C^\star))\ge 10\lambda^2\eps\OPT_\beta(X)$. $\Pr[\neg\phi_S]$ can be regarded as a uniformly convergence guarantee for all $v^C\in V$ and, let $\mu:=\frac{\eps m}{n}\OPT_\beta(X)$, it is equivalent to bound
\begin{equation}
    \label{eq:uniform_convergence}
    \Pr_{S}\left[\inf_{C\in\C^\mathrm{bad}(X)} \sum_{x\in S}v^C_x\le \lambda^2 \mu\right]\le 0.09
\end{equation}
To bound~\eqref{eq:uniform_convergence},
our plan is to set up a Gaussian process and apply the chaining argument.
The following lemma provides a convergence guarantee for a single $C$.
\begin{lemma}
    \label{lem:badsol}
    For any $C\in\C^\mathrm{bad}(X)$, the following holds:  
    \begin{equation*}
        \Pr_S\left[\sum_{x\in S}v_x^C\le 5\lambda^2\mu \right]\le 2\exp\left(-\Theta\left(\frac{\epsilon^2\beta m}{k} \right) \right).
    \end{equation*}
\end{lemma}

\begin{proof}
	Since $\C^\mathrm{bad}(X)\subseteq \C'(X)$, due to~\eqref{eq:b2}, for every $C\in\C^\mathrm{bad}(X)$, it holds that 
	\begin{equation*}
	\|v^C\|_\infty\le \gamma=\frac{12\lambda k}{\beta n}\OPT_\beta(X).
	\end{equation*}
	Since $S$ is a set of $m$ uniform samples, and $\E\sum_{x\in S}v_x^C=\frac{m}{n}\|v^C\|_1\ge \frac{10\lambda^2\epsilon m}{n}\OPT_\beta(X)$, we have 
	\begin{align*}
	\Pr_S\left[\sum_{x\in S}v^C_x\le \frac{5\lambda^2\epsilon m}{n}\OPT_\beta(X) \right]\le \Pr_S\left[\left|\sum_{x\in S}v^C_x-\E\sum_{x\in S}v^C_x \right|\ge \frac{5\lambda^2\epsilon m}{n}\OPT_\beta(X) \right]
	\end{align*}
	Here we apply Bernstein inequality to finish the proof. To this end, we should also bound the variance of $\sum_{x\in S}v^C_x$, which is
	\begin{align*}
	&\le m\cdot\frac{1}{n}\sum_{x\in X}(v^C_x)^2\\
	&\le \gamma \cdot \frac{m}{n}\sum_{x\in X} |v^C_x|\\
	&\le \gamma \cdot \frac{m}{n}\sum_{x\in X}\left(\dist(x,C)+\dist(x,C^\star) \right)\\
	&\le \gamma \cdot \frac{m}{n}\sum_{x\in X}\left(2\dist(x,C^\star)+\dist(C^\star(x),C) \right)\\
	&\le O\left(\frac{km}{\beta n^2}(\OPT_\beta(X))^2 \right).
	\end{align*}
	where the last derivation is due to~\eqref{eq:b1}.
	Therefore, by Bernstein inequality, we have 
	\begin{align*}
	\Pr_S\left[\left|\sum_{x\in S}v^C_x-\E\sum_{x\in S}v^C_x \right|\ge \frac{5\lambda^2\epsilon m}{n}\OPT_\beta(X) \right]\le 2\exp\left(-\Theta\left(\frac{\eps^2\beta m}{k} \right) \right)
	\end{align*}
	which completes the proof.
\end{proof}

By~\Cref{lem:badsol}, it remains to bound the ``complexity'' of $\C^\mathrm{bad}(X)$.
To this end, we reduce to a Gaussian process where we only need to consider the complexity of coverings with respect to the sample set $S$. This idea is formalized in the following lemma.

\begin{lemma}[Reduction to Gaussian Process]
    \label{lem:reduce_to_Guassian}
    ~\eqref{eq:uniform_convergence} holds if the following holds:
    \begin{equation}
        \label{eq:Guassian}
        \E_S\left[\E_{g_i}\left[\sup_{C\in \C^\mathrm{bad}(X)} \frac{1}{\mu}\left|\sum_{i=1}^mg_iu^C_{s_i} \right| \right]\mid \xi_S\right] \le \lambda.
    \end{equation}
Here, $g_1,\dots,g_m$ are the independent standard Gaussian random variables, $U_{\eps}$ is an $\eps$-covering of $V$ w.r.t. a random set $S=\{s_1,\dots,s_m\}\subseteq X$, and $u^C\in U_{\eps}$ denotes the $\eps$-covering of $v^C$ for $C\in\C^\mathrm{bad}(X)$,
\end{lemma}

\begin{proof}
	Let $S'=\{s_1',...,s_m'\}$ be another set of independent uniform samples from $X$ that is independent of $S$. We first use the symmetrization trick to show that
	\begin{align*}
	\quad\Pr_S\left[\inf_{C\in\C^\mathrm{bad}(X)} \sum_{i=1}^m v^C_{s_i}\le \lambda^2\mu\right] \nonumber\le 2\Pr_{S,S'}\left[\sup_{C\in\C^\mathrm{bad}(X)} \left|\sum_{i=1}^m(v_{s_i}^C-v_{s_i'}^C) \right|\ge 4\lambda^2\mu \right].
	\end{align*}
	To see this, we assume for some $S$, the event $\phi_S$, i.e., $\inf_{C\in\C^\mathrm{bad}(X)} \sum_{i=1}^m v^C_{s_i}\le \lambda^2\mu$ happens, then take an arbitrarily $C_S\in \C^\mathrm{bad}(X)$ such that $\sum_{i=1}^m v^{C_S}_{s_i}\le \lambda^2\mu$ (if $\xi_S$ does not happen, we let $C_S$ be an arbitrarily center set). If the event $\sum_{i=1}^m v^{C_S}_{s_i'}\ge 5\lambda^2\mu$, denoted by $\varphi_{C_S,S'}$, happens, then it holds that
	\begin{equation*}
	\left|\sum_{i=1}^m(v^{C_S}_{s_i}-v^{C_S}_{s_i'})\right|\ge 4\lambda^2\mu.
	\end{equation*}
	Note that $\varphi_{C_S,S'}$ is a convergence guarantee for $C_S\in\C^\mathrm{bad}(X)$, and~\Cref{lem:badsol} gives a lower bound of $1/2$ to $\Pr_{S'}[\varphi_{C_S,S'}]$.
	Therefore, the following holds.
	\begin{align*}
	\Pr_{S,S'}\left[\sup_{C\in\C^\mathrm{bad}(X)} \left|\sum_{i=1}^m(v_{s_i}^C-v_{s_i'}^C) \right|\ge 4\lambda^2\mu \right]
	&\ge \Pr_{S,S'}\left[\phi_S\wedge \varphi_{C_S,S'}\right]\\
	&=\Pr_{S}[\phi_S]\Pr_{S,S'}[\varphi_{C_S,S'}\mid \phi_S]\\
	&=\Pr_S[\phi_S]\E_S\left[\Pr_{S'}[\varphi_{C_S,S'}]\mid\phi_S\right]\\
	&\ge \frac{1}{2}\Pr_S[\phi_S],
	\end{align*}
	Let $r_1,...,r_m$ be independent Rademacher random variables\footnote{A Rademacher random variable $r$ takes value $-1$ with probability $1/2$ and takes value $1$ with probability $1/2$.}. We have
	\begin{align}
	&\quad\Pr_{S,S'}\left[\sup_{C\in\C^\mathrm{bad}(X)} \left|\sum_{i=1}^m(v_{s_i}^C-v_{s_i'}^C) \right|\ge 4\lambda^2\mu \right] \nonumber\\
	&= \Pr_{S,S',r_i}\left[\sup_{C\in\C^\mathrm{bad}(X)} \left|\sum_{i=1}^m r_i(v_{s_i}^C-v_{s_i'}^C) \right|\ge 4\lambda^2\mu \right] \nonumber\\
	&\le \Pr_{S,S',r_i}\left[\sup_{C\in\C^\mathrm{bad}(X)} \left(\left|\sum_{i=1}^m r_iv_{s_i}^C \right|+\left|\sum_{i=1}^m r_i v_{s_i'}^C \right|\right)\ge 4\lambda^2\mu \right] \nonumber
	\end{align}
	where the first derivation is because $v_{s_i}^C-v_{s_i'}^C$ is symmetric and thus is distributed identically to $r_i(v_{s_i}^C-v_{s_i'}^C)$, the second derivation is due to the triangle inequality. If $\sup_{C\in\C^\mathrm{bad}(X)} \left(\left|\sum_{i=1}^m r_iv_{s_i}^C \right|+\left|\sum_{i=1}^m r_i v_{s_i'}^C \right|\right)\ge 4\lambda^2\mu $ holds, then either $\sup_{C\in\C^\mathrm{bad}(X)} \left|\sum_{i=1}^m r_iv_{s_i}^C \right|\ge 2\lambda^2\mu$ or $\sup_{C\in\C^\mathrm{bad}(X)} \left|\sum_{i=1}^m r_iv_{s_i'}^C \right|\ge 2\lambda^2\mu$ holds. Since $S$ and $S'$ are distributed identically, by union bound, we have 
	\begin{align}
		&\quad\Pr_{S,S',r_i}\left[\sup_{C\in\C^\mathrm{bad}(X)} \left(\left|\sum_{i=1}^m r_iv_{s_i}^C \right|+\left|\sum_{i=1}^m r_i v_{s_i'}^C \right|\right)\ge 4\lambda^2\mu \right] \nonumber\\
		&\le 2\Pr_{S,r_i}\left[\sup_{C\in\C^\mathrm{bad}(X)} \left|\sum_{i=1}^m r_iv_{s_i}^C\right|\ge 2\lambda^2\mu \right] \nonumber\\
		&= 2\Pr_{S,r_i}\left[\sup_{C\in\C^\mathrm{bad}(X)} \left|\sum_{i=1}^m r_iv_{s_i}^C\right|\ge 2\lambda^2\mu\mid \xi_S \right]\Pr[\xi_S]\nonumber\\
		&\quad+2\Pr_{S,r_i}\left[\sup_{C\in\C^\mathrm{bad}(X)} \left|\sum_{i=1}^m r_iv_{s_i}^C\right|\ge 2\lambda^2\mu\mid \neg\xi_S \right]\Pr[\neg \xi_S]\nonumber\\
		&\le 2\Pr_{S,r_i}\left[\sup_{C\in\C^\mathrm{bad}(X)} \left|\sum_{i=1}^m r_iv_{s_i}^C\right|\ge 2\lambda^2\mu\mid \xi_S \right] + 0.02, \label{eq:pr}
	\end{align}
	It suffices to prove that 
	\begin{equation}\label{eq:ES}
	\E_{S}\left[\E_{r_i}\sup_{C\in\C^\mathrm{bad}(X)}\left|\sum_{i=1}^m r_iv_{s_i}^C \right|\mid \xi_S \right]\le 20\lambda\mu.
	\end{equation}
	We can thus apply Markov inequality to bound the probability in~\eqref{eq:pr} by $0.01$, which leads to $\Pr[\neg \phi_S]\le 0.08$ and completes the proof of~\eqref{eq:uniform_convergence}.
	It remains to show how to derive~\eqref{eq:ES}. 
	
	Recall that $U_{\eps}$ denotes an $\eps$-covering of $V$ w.r.t. random subset $S$.
	We next replace the cost vector $v^C$ with its $\eps$-covering $u^C\in U_{\eps}$, which satisfies that $|u^C_x-v^C_x|\le \eps\cdot \err(x,v^C)$ for every $x\in S$.
\begin{equation}    
	\label{eq:applycover}
	\begin{aligned}
	\sup_{C\in\C^\mathrm{bad}(X)}\left|\sum_{i=1}^mr_iv_{s_i}^C \right|
	&\le \sup_{C\in\C^\mathrm{bad}(X)} \left|\sum_{i=1}^mr_i u_{s_i}^C +r_i(v_{s_i}^C- u_{s_i}^C) \right|\\
	&\le \sup_{C\in\C^\mathrm{bad}(X)}\left(\left|\sum_{i=1}^mr_i u^C_{s_i} \right| + \eps \sum_{i=1}^m\err(s_i,v^C)\right)\\
	&\le \sup_{C\in\C^\mathrm{bad}(X)}\left|\sum_{i=1}^mr_iu^C_{s_i} \right| + \eps\cdot\frac{15\lambda m}{n}\OPT_\beta(X)\\
	&\le \sup_{C\in\C^\mathrm{bad}(X)}\left|\sum_{i=1}^mr_iu^C_{s_i} \right|+15\lambda\mu
	\end{aligned}
	\end{equation}
	where the second derivation is due to the triangle inequality, and the third derivation is due to \Cref{lem:bounderror} and $|r_i|=1$ for every $i\in[m]$. Then it suffices to prove
	\begin{equation*} 
		\E_{S}\left[\E_{r_i}\sup_{C\in\C^\mathrm{bad}(X)}\left|\sum_{i=1}^m r_iu_{s_i}^C \right|\mid \xi_S \right]\le 5\lambda\mu
	\end{equation*}
	which is equivalent to
	\begin{equation}
		\label{eq:Eg}
		\E_{S}\left[\E_{r_i}\sup_{C\in\C^\mathrm{bad}(X)}\frac{1}{\mu} \left|\sum_{i=1}^m r_iu_{s_i}^C \right|\mid \xi_S \right]\le 5\lambda
	\end{equation}
	
	Finally, we replace the Rademacher random variables with standard Gaussian random variables.
	\begin{lemma}[Lemma 7.4 of~\cite{Handel2014ProbabilityIH}]
		\label{lem:r2g}
		For $r_1,\dots,r_m$ are Rademacher random variables, let $g_1,\dots,g_m$ be the independent standard Gaussian random variables, it holds that
		\begin{equation*}
		\E_{r_i}\left[\sup_{C\in\C^\mathrm{bad}(X)} \frac{1}{\mu}\left|\sum_{i=1}^mr_iu^C_{s_i} \right| \right]\le \sqrt{\frac{\pi}{2}}\E_{g_i}\left[\sup_{C\in\C^\mathrm{bad}(X)} \frac{1}{\mu}\left|\sum_{i=1}^mg_iu^C_{s_i} \right| \right]
		\end{equation*}
	\end{lemma}
	It suffices to prove 
	\begin{equation*}
		\E_{S}\left[\E_{g_i}\sup_{C\in\C^\mathrm{bad}(X)}\frac{1}{\mu} \left|\sum_{i=1}^m g_iu_{s_i}^C \right|\mid \xi_S \right]\le \lambda
	\end{equation*}
	which leads to \eqref{eq:Eg} by \Cref{lem:r2g}, and completes the proof.
\end{proof}

Now it suffices to prove~\eqref{eq:Guassian}.
The main idea is to apply a chaining argument.
For every $C\in \C^\mathrm{bad}(X)$, let $v^{C,h}\in U_{2^{-h}}$ denote the $2^{-h}$-covering of $v^C$, then we can rewrite $u^C$ as a telescoping sum 
$
    u^C=\sum_{h=1}^{\log\eps^{-1}}(v^{C,h}-v^{C,h-1}).
$
Hence, it suffices to prove the following lemma which implies~\eqref{eq:Guassian}, and this completes the proof of~\Cref{thm:dim_ind}.

\begin{lemma}[Bounding Error in the Chaining Argument]
	\label{lem:chaining_argument}
    Conditioning on $\xi_S$, the following holds:
	\begin{equation*}
	\sum_{h=1}^{\log\eps^{-1}}\E_{g_i}\left[\sup_{C\in\C^\mathrm{bad}(X)}\frac{1}{\mu}\left|\sum_{i=1}^mg_i(v_{s_i}^{C,h}-v_{s_i}^{C,h-1}) \right| \right] \le \lambda.
	\end{equation*}
\end{lemma}

The proof of \Cref{lem:chaining_argument} heavily relies on our definition of covering. 
In particular, it allows us to bound the difference $|v_x^{C,h}-v_x^{C,h-1}|$
either by $2^{-h+2}\cdot\err(x,v^C)$, or by an absolute value $2\gamma$. 
This eventually guarantees that each Gaussian variable $\sum_{i=1}^mg_i(v_{s_i}^{C,h}-v_{s_i}^{C,h-1}) $ has a well bounded variance.

\begin{proof}[Proof of~\Cref{lem:chaining_argument}]

For every $h\in [\log\eps^{-1}]$, let 
	\begin{equation*}
		E_h:=\E_{g_i}\left[\sup_{C\in\C^\mathrm{bad}(X)}\frac{1}{\mu}\left|\sum_{i=1}^mg_i(v_{s_i}^{C,h}-v_{s_i}^{C,h-1}) \right| \right].
	\end{equation*}
	For every $C\in \C^\mathrm{bad}(X)$, $\sum_{i=1}^m\frac{g_i}{\mu}(v_{s_i}^{C,h}-v_{s_i}^{C,h-1})$ is Gaussian with zero mean and variance
	\begin{align*}
	\sum_{i=1}^m\left(\frac{1}{\mu}(v_{s_i}^{C,h}-v_{s_i}^{C,h-1}) \right)^2
	&\le \frac{\|v^{C,h}\|_\infty+\|v^{C,h-1}\|_\infty }{\mu}\sum_{i=1}^m\frac{|v_{s_i}^{C,h}-v_{s_i}^{C,h-1}|}{\mu}\\
	&\le \frac{24\lambda k}{\beta\epsilon m}\sum_{i=1}^m\frac{|v_{s_i}^{C,h}-v^{C}_{s_i}+v^{C}_{s_i}-v_{s_i}^{C,h-1} |}{\mu}\\
	&\le \frac{24\lambda k}{\beta\epsilon m}\sum_{x\in S}\frac{2^{-h+2}\err(x,v^{C})}{\mu}\\
	&\le O\left(\frac{2^{-h+6} k}{\beta\epsilon^2 m}\right)\\
	\end{align*}
	where the second derivation is due to $\|v^{C,h}\|_\infty\le \gamma$ for all $h\in [\log\eps^{-1}]$ and the last derivation is due to~\Cref{lem:bounderror}. The following lemma demonstrates that an upper bound of the variance of each $\sum_{i=1}^m\frac{g_i}{\mu}(v_{s_i}^{C,h}-v_{s_i}^{C,h-1})$ leads to an upper bound of $E_h$.
	
	\begin{lemma}[Lemma 2.3 of~\cite{massart2007concentration}]
		Let $g_i\sim N(0,\sigma_i^2)$ for $i\in [m]$ be Gaussian random variables (not need to be independent ) and let $\sigma=\max_{i\in m}\sigma_i$, then it holds that
		\begin{equation*}
		\E\left[\max_{i\in [m]}|g_i| \right]\le 2\sigma\cdot\sqrt{2\ln n}.
		\end{equation*} 
	\end{lemma}
	The number of distinct difference vector $v^{C,h}-v^{C,h-1}$ is at most $|U_{2^{-h}}\times U_{2^{-h+1}}|\le N_X^{2^{-h}}(m)\cdot N_X^{2^{-h+1}}(m)$.
	Therefore, we have 
	\begin{align*}
		E_h&\le \sqrt{8\log (|U_{2^{-h}}|\cdot|U_{2^{-h+1}}|)\cdot O\left(\frac{2^{-h+6} k}{\beta\epsilon^2 m}\right)}\\
	&\le \sqrt{O(1)\cdot \log N_X^{2^{-h}}(m)\cdot \frac{2^{-h}k}{\beta\eps^2m}}
	\end{align*}
	Plug in $\sum_{h=1}^{\log\eps^{-1}}E_h$, we have 
	\begin{equation}
	\label{eq:final}
	\sum_{h=1}^{\log \eps^{-1}}E_h\le O(1)\cdot \sqrt{\frac{k}{\beta\eps^2m}}\sum_{h=1}^{\log\eps^{-1}} \sqrt{\log N_X^{2^{-h}}(m)}.
	\end{equation}
	Since $m\ge t\cdot \frac{k}{\beta\eps^2}\left(\sum_{i=1}^{\log\eps^{-1}}\sqrt{2^{-i}\log N^{2^{-i}}_X(m)} \right)^2$ for sufficiently large constant $t$, we can bound $\sum_{h=1}^{\log\eps^{-1}}E_h$ by $\lambda$.
	which completes the proof.
\end{proof}

\section{Weak Coresets in Various Metric Spaces}
\label{sec:application}
We apply \Cref{thm:dim_ind} to various metric spaces and obtain the following theorem, by analyzing their covering number.

\begin{theorem}
    \label{thm:various_metric}
    For a metric space $M=(\calX,\dist)$ and a dataset $X\subseteq \calX$, an integer $k\ge 1$ and real numbers $\beta,\epsilon\in(0,1)$, let $S$ be a set of uniform samples with size
    \begin{itemize}
        \item $O\left(\frac{k^2}{\beta\eps^3}\cdot\log^2\frac{k}{\beta\eps}\cdot\log^2\frac{1}{\eps} \right)$ if $M$ is Euclidean $\mathbb{R}^d$;
        \item $O\left(\frac{k^2}{\beta\eps^2}\cdot \ddim\cdot \log\frac{k}{\beta\eps}\right)$ if $M$ has doubling dimension $\ddim$;
        \item $O\left(\frac{k^2}{\beta\eps^2}\cdot \log |\calX|\cdot \log\frac{k}{\beta\eps}\right)$ if $M$ is a finite metric;
        \item $O\left(\frac{k^2}{\beta\eps^2}\cdot \tw\cdot \log\frac{k}{\beta\eps}\right)$ if $M$ is the shortest-path metric of a graph with treewidth $\tw$.
    \end{itemize}
    Then $S$ is an $\epsilon$-weak coreset for \kBMedian on $X$ with probability $0.9$.
\end{theorem}

\subsection{Euclidean Space}
\begin{lemma}[Covering Number in Euclidean Space]
    \label{lem:covering_euclidean}
    In Euclidean space $(\R^d,\dist)$, for integer $m>0$ and real number $0<\alpha<1/2$, it holds that  
    \begin{equation*}
        \log |N_X^{\alpha}(m)|\le O\left(k \alpha^{-2}\log(m+k)\log\frac{k}{\beta\alpha}\right)
    \end{equation*}
\end{lemma}
\begin{proof}
By \Cref{def:covering} of the covering number, it suffices to construct an $\alpha$-covering of $V$ w.r.t $(S,\err)$ for every $S\in \mathcal{S}(m)$.
We need \emph{terminal embedding} as stated in the following theorem:
\begin{theorem}[Terminal Johnson-Lindenstrauss Lemma \cite{DBLP:conf/stoc/NarayananN19}] 
    \label{thm:terminal}
    For every $\varepsilon\in(0,1/2)$ and finite set $Y\subseteq \R^d$, there exists an embedding $g:\R^d\to\R^t$ for $t=O(\varepsilon^{-2}\log |Y|)$ such that 
    \begin{eqnarray*}
        \forall x\in Y,\forall c\in \R^d,\quad \dist(x,c)\le \dist(g(x),g(c))\le (1+\varepsilon)\dist(x,c).
    \end{eqnarray*}
    We call $g$ is an $\varepsilon$-terminal embedding for $Y$.
\end{theorem}

For every $S\in \mathcal{S}(m)$, let $Y:=S\cup C^\star$, by \Cref{thm:terminal}, there exists an $\alpha$-terminal embedding $g$ for $Y$ with target dimension $t=O(\alpha^{-2}\log |Y|)=O(\alpha^{-2}\log (m+k))$. For a set $A\subset \R^d$, we denote by $g(A):=\{g(x):x\in A\}$ the set of images of all $x\in A$.
By the definition of $\C'(X)$, for every $C\in \C'(X)$, $c\in C$,
\begin{align*}
    \dist(g(c),g(C^\star))&\le \dist(g(c),g(C^\star(c)))\\
    &\le 2\dist(c,C^\star)\\
    &\le \frac{12\lambda k}{\beta n}\OPT(X)
\end{align*} 
We denote by $B_A(c,r):=\{x\in A:\dist(x,c)\le r\}$ the ball centered at $c$ of radius $r$ for $A\subseteq \calX$. Thus we have for every $C\in \C'(X)$, 
\[
    g(C)\subseteq \bigcup_{c^\star\in C^\star}B_{\R^t}\left(g(c^\star), \frac{12\lambda k}{\beta n}\OPT_\beta(X)\right)
\]
To construct a covering for the set of cost vectors, we discretize each ball via a classical \emph{covering} of a point set.

\begin{definition}[Covering of a Point Set]
    For a metric space $(\calX,\dist)$, a point set $A\subseteq\calX$ and real number $0\le\alpha<1$, we say $T\subseteq\calX$ is an $\alpha$-covering of $A$ if for every $x\in A$, there exists $y\in T$ such that $\dist(x,y)\le \alpha$.
\end{definition}
Note that the above definition of covering is different from \Cref{def:covering}, since they are for different objects. The following lemma bounds the cardinality of $\alpha$-covering of an Euclidean ball.
\begin{lemma}[Covering of a Euclidean Ball]
    \label{lem:ball_cover}
    For $\alpha>0$ and an Euclidean ball $B\subset \R^t$ of radius $\Delta>0$, there exists an $\alpha$-covering $T\subseteq B$ of size at most $\exp\left({O(t\log (\Delta/\alpha))}\right)$.
\end{lemma}

For every $c^\star\in C^\star$, let $T_{c^\star}$ be such an $\frac{\alpha\OPT_\beta(X)}{n}$-covering of $B_{\R^t}\left(g(c^\star), \frac{12\lambda k}{\beta} \cdot\frac{1}{n}\OPT_\beta(X)\right)$, and let $T_{C^\star}:=\bigcup_{c^\star\in C^\star}T_{c^\star}$. By \Cref{lem:ball_cover}, we have 
\begin{eqnarray*}
    |T_{C^\star}|\le k\cdot \exp\left({O\left(t\log{\frac{k}{\alpha\beta}}\right)}\right)
\end{eqnarray*}
Let $g':\R^d\to T_{C^\star}$ be a function satisfying that
\begin{eqnarray*}
    g'(x)=\arg \min_{y\in T_{C^\star}}\dist(g(x),y),    
\end{eqnarray*}
Construct $U$ is $\{\tilde v^{C}\in \R^{X}:C\in\C'(X)\}$, where $\tilde v^{C}$ is a cost function defined as: for every $x\in X$, $\tilde v^{C}_x:=\dist(g(x),g'(C))-\dist(g(x),g(C^\star))$. Observe that 
\begin{equation*}
    |U|\le |T_{C^\star}|^k\le \exp\left({O\left(kt\log\frac{k}{\alpha\beta}\right)}\right),
\end{equation*}
which implies that $\log|U|\le O(kt\log\frac{k}{\beta\alpha})=O(k\alpha^{-2}\log(m+k)\log\frac{k}{\beta\alpha})$, it remains to show that $U$ is the desired $\alpha$-covering of $V$.
\paragraph{Bounded Covering Error}
By the definition of $g'$, we have that for every $C\in \C'(X)$ and $c\in C$, 
\begin{equation*}
    \dist(g(c),g'(C)) \le \dist(g(c),g'(c))\le \frac{\alpha}{n}\OPT_\beta(X),
\end{equation*}
and
\begin{equation*}
    \dist(g'(c),g(C))\le \dist(g'(c),g(c))\le \frac{\alpha}{n}\OPT_\beta(X),
\end{equation*}
and thus for every $x\in S$, let $c_1$ denote the closest center in $g(C)$ to $g(x)$, and $c_2$ denote the closest center in $g'(C)$ to $g(x)$, then it holds that
\begin{equation}
    \label{eq:error_ball}
    |\dist(g(x),g(C))-\dist(g(x),g'(C))|\le \max\left\{\dist(c_1,g'(C)),\dist(c_2,g(C))\right\}\le \frac{\alpha}{n}\OPT_\beta(X).
\end{equation}
The error incurred by the covering is
\begin{align*}
|v^C_x-\tilde v^C_x|&=|\dist(x,C)-\dist(x,C^\star)-\dist(g(x),g'(C))+\dist(g(x),g(C^\star)) |\\
&\le |\dist(x,C)-\dist(g(x),g'(C)) | + |\dist(x,C^\star)-\dist(g(x),g(C^\star)) |\\
&\le |\dist(x,C)-\dist(g(x),g(C)) | + |\dist(x,C^\star)-\dist(g(x),g(C^\star)) | +\frac{\alpha}{n}\OPT_\beta(X) \\
&\le \alpha\dist(x,C)+\alpha\dist(x,C^\star)+\frac{\alpha}{n}\OPT_\beta(X)\\
&=\alpha\cdot \err(x,v^C)
\end{align*}
where the second derivation is due to triangle inequality, the third derivation is due to~\eqref{eq:error_ball}, and the forth derivation is due to the terminal embedding guarantee.
\paragraph{Bounded $L_\infty$ Norm} For every $\tilde v^C\in U$,
\begin{align*}
    \|\tilde v^C\|_\infty&=\max_{x\in X}|\dist(g(x),g'(C))-\dist(g(x),g(C^\star))|\\
    &\le\max_{x\in X} |\dist(g(x),g(C))-\dist(g(x),g(C^\star))| + \frac{\alpha}{n}\OPT_\beta(X)\\
    &\le \max\left\{\max_{c\in C}\dist(g(c),g(C^\star)),\max_{c^\star\in C^\star}\dist(g(c^\star),g(C)) \right\}+ \frac{\alpha}{n}\OPT_\beta(X)\\
    &\le (1+\alpha)\max\left\{\max_{c\in C}\dist(c,C^\star),\max_{c^\star\in C^\star}\dist(c^\star,C) \right\}+ \frac{\alpha}{n}\OPT_\beta(X)\\
    &\le \frac{12\lambda k}{\beta n}\OPT_\beta(X),
\end{align*}
where the second derivation is due to~\eqref{eq:error_ball}, the third derivation is due to the triangle inequality, and the last derivation follows from a similar argument of \Cref{lem:balance_is_good}.
Thus we can construct $\alpha$-covering of $V$ with respect to any $S\in \mathcal{S}(m)$, which concludes the proof.
\end{proof}
\begin{proof}[Proof of \Cref{thm:various_metric} in Euclidean space]
    We only need to bound the summation in \eqref{eq:m_X}.
    \begin{align*}
        &\quad\sum_{i=1}^{\log\eps^{-1}}\sqrt{2^{-i}\log N^{2^{-i}}_X(m)}\\
        &\le \sum_{i=1}^{\log\eps^{-1}} O(1)\cdot  \sqrt{2^{-i}\cdot k\cdot 2^{2i}\log m\log\frac{k}{\beta 2^{-i}}}\\
        &\le O(1)\cdot \log\eps^{-1}\cdot \sqrt{\eps^{-1}\cdot k \log m\log \frac{k}{\beta\eps}}.
    \end{align*}
    Therefore, it suffices to set $m=O\left(\frac{k^2}{\beta\eps^3}\cdot\log^2\frac{k}{\beta\eps}\cdot\log^2\frac{1}{\eps} \right)$.
\end{proof}

\subsection{Doubling Metric Space}
\begin{definition}[Doubling Dimension~\cite{assouad1983plongements,DBLP:conf/focs/GuptaKL03}]
	The doubling dimension of a metric space $M = (\calX, \dist)$, denoted as $\ddim(M)$, is the smallest integer $d$ such that any ball can be covered by at most $2^d$ balls of half the radius.
\end{definition}
\begin{lemma}[Covering Number in Doubling Metric Space]
    \label{lem:covering_doubling_metric}
    In a metric space $(\calX,\dist)$ with doubling dimension $\ddim$, for integer $m>0$ and real number $0<\alpha<1/2$, it holds that  
    \begin{equation*}
        \log |N_X^{\alpha}(m)|\le O(k\cdot \ddim \cdot \log (k/\alpha\beta)).
    \end{equation*}
\end{lemma}
\begin{proof}
    It suffices to construct $\alpha$-covering with respect to $S$ for every $S\in \mathcal{S}(m)$. Different from the Euclidean case (\Cref{lem:covering_euclidean}), we directly apply the covering of point set without using terminal embedding. 
    The following lemma bounds the cardinality of a covering of point set in doubling metrics.
    \begin{lemma}[\cite{DBLP:conf/focs/GuptaKL03}]
        \label{lem:cover_dm}
        For $\alpha>0$ and a metric space $(\calX,\dist)$ with doubling dimension $\ddim$ and diameter $\Delta$, there exists an $\alpha$-covering $T$ of $\calX$ with $|T|\le \exp\left(O(\ddim\cdot\log (\Delta/\alpha))\right)$.
    \end{lemma}
    
    By definition of $\C'(X)$, for every $C\in\C'(X)$, $c\in C$, it holds that $\dist(c,C^\star)\le \frac{6\lambda k}{\beta}\cdot \frac{1}{n}\OPT(X)$. Hence, consider an $\frac{\alpha\OPT(X)}{n}$-covering $T_{c^\star}$ of $B_{\calX}\left(c^\star, \frac{6\lambda k}{\beta} \cdot\frac{1}{n}\OPT(X)\right)$ for every $c^\star\in C^\star$ and let $T_{C^\star}:=\bigcup_{c^\star\in C^\star}T_{c^\star}$. By \Cref{lem:cover_dm}, we have 
    \begin{eqnarray*}
        |T_{C^\star}|\le k\cdot \exp\left(O\left(\ddim\cdot \log{\frac{k}{\alpha\beta}}\right)\right)
    \end{eqnarray*}
    We define $g:V\to V$ such that for every $c\in V$,
    \begin{equation*}
        g(c)=\arg\min_{c'\in T_{C^\star}} \dist(c,c').
    \end{equation*}
    Construct $U$ is $\{\tilde v^{C}\in \R^{X}:C\in\C'(X)\}$, where $\tilde v^{C}$ is a cost function defined as: for every $x\in X$, $\tilde v^{C}_x:=\dist(x,g(C))-\dist(x,C^\star)$. Then a similar analysis for bounding the covering error and $L_\infty$ norm as in Euclidean case(\Cref{lem:covering_euclidean}) certifies that $U$ is an $\alpha$-covering of $V$ with respect to any $S\in \mathcal{S}(m)$, and $\log |U|\le O(k\cdot\ddim\cdot\log\frac{k}{\alpha\beta})$.
\end{proof}
\begin{proof}[Proof of \Cref{thm:various_metric} in doubling metric space and general discrete metric space]
    Directly plugging the covering number in \Cref{lem:covering_doubling_metric} into \eqref{eq:simple_m_X} concludes the proof in doubling metric space. For general discrete metric space $M=(\calX,\dist)$, we know that the doubling dimension is $O(\log |\calX|)$, and thus directly applying the result of doubling metric space completes the proof.
\end{proof}

\subsection{Shortest-path Metric of a Graph with Bounded Treewidth}
\begin{definition}[Tree Decomposition and Treewidth]
	A tree decomposition of a graph $G=(V, E)$ is a tree $\mathcal{T} = (\mathcal{V}, \mathcal{E})$, where each node in $\mathcal{V}$, called a bag, is a subset of vertices in $V$, such that the following holds.
	\begin{itemize}
		\item $\bigcup_{S \in \mathcal{V}} S = V$.
		\item $\forall u \in V$, the nodes in $\mathcal{V}$ that contain $u$ form a connected component in $\mathcal{T}$.
		\item $\forall (u, v) \in E$, there exists $S \in \mathcal{V}$ such that $\{u, v\} \subseteq S$.
	\end{itemize}
	The treewidth of $G$, denoted as $\tw(G)$, is the smallest integer $t$ such that there is a tree decomposition of $G$ with maximum bag size $t + 1$.
\end{definition}

\begin{lemma}[Covering Number in Shortest-path Metric of a Graph]
    \label{lem:covering_graph_metric}
    In shortest-path metric $M=(\calX,\dist)$ of a graph with bounded treewidth $\tw$, for integer $m>0$ and real number $0<\alpha<1/2$, it holds that  
    \begin{equation*}
        \log |N_X^{\alpha}(m)|\le O(k\cdot \tw \cdot \log (k/\alpha\beta + m/\alpha)).
    \end{equation*}
\end{lemma}
\begin{proof} 
    To bound the covering number in graph metrics, it suffices to construct an $\alpha$-covering of $V$ with respect $S$ for every $S\in \mathcal{S}(m)$.  
    The proof of \Cref{lem:covering_graph_metric} relies on the following structural lemma, proposed by~\cite{DBLP:conf/icml/BakerBHJK020}.
    \begin{lemma}[Structural Lemma, \cite{DBLP:conf/icml/BakerBHJK020}]
        \label{lem:structral_lem}
        Given graph $G=(\calX,E)$ with treewidth $\tw$, and $A\subseteq \calX$, there exists a collection $\mathcal{T}_A$ of subsets of $\calX$, such that the following holds.
        \begin{itemize}
            \item[1.] $\bigcup_{T\in \mathcal{T}_A} T=\calX$.
            \item[2.] $|\mathcal{T}_A|\le \poly(|A|)$.
            \item[3.] For every $T\in\mathcal{T}_A$, either $|T|\le O(\tw)$, or i) $|T\cap A|\le O(\tw)$ and ii) there exists $P_T\subseteq V$ with $|P_T|\le O(\tw)$ such that there is no edge in $E$ between $T$ and $\calX\setminus (T\cup P_T)$.
        \end{itemize}
    \end{lemma}
    Recall that any $S\in \mathcal{S}(m)$ makes $\xi_S$ happens, i.e., $S$ satisfies \eqref{eq:xi_S}. Thus for every $C\in \C'(X)$ and every $x\in S$, 
    \begin{eqnarray*}
        \dist(x,C^\star)\le \sum_{x\in S}\dist(x,C^\star)\le \frac{\lambda m}{n}\OPT(X),
    \end{eqnarray*}
    and thus 
    \begin{equation*}
        \dist(x,C)\le \dist(C(x),C^\star)+\dist(x,C^\star)\le\left(\frac{6k}{\beta}+m \right)\cdot\frac{\lambda}{n}\OPT(X).
    \end{equation*}
    Let $\mathcal{T}_S$ be a collection of subsets asserted by \Cref{lem:structral_lem}. 
    We can construct the covering as follows:
    
    For every $T\in\mathcal{T}$, and $c\in T$, let $P_{T}\subseteq \calX$ be a set asserted in condition 3 in \Cref{lem:structral_lem} (we assume $|T|>O(\tw)$, otherwise let $P_{T}=T$). By \Cref{lem:structral_lem}, we have for every $x\in \calX\setminus T$, it holds that
    \begin{equation*}
        \dist(x,c)=\min_{y\in P_{T}}\left\{\dist(x,y) + \dist(y,c)\right\},
    \end{equation*}
    since the shortest path between $x$ and $c$ must pass by some vertices in $P_T$.
    Let $I_T:=(S\cap T)\cup P_T$ denote the \emph{important vertices} of $T$,
    we define a \emph{rounded distance function} $\dist_T':T\times X\to \R_+$ w.r.t. $T$ as follows:
    \begin{itemize}
        \item $\forall c\in T, x\in I_T$, $\dist_T'(c,x)$ is the closest multiple of $\frac{\alpha}{n}\OPT(X)$ to $\dist(x,c)$ no greater than $\left(\frac{6k}{\beta}+m \right)\cdot\frac{\lambda}{n}\OPT(X)$.
        \item $\forall c\in T, x\in X\setminus I_T$, $\dist'_T(c,x)=\min_{y\in P_T}\left\{\dist'_T(c,y)+\dist(y,x) \right\}$.
    \end{itemize}
    Note that for $x\in T\setminus S$ or for $x\in \calX$ that satisfies $\dist(c,x)>\left(\frac{6k}{\beta}+m \right)\cdot\frac{\lambda}{n}\OPT(X)$, the rounded distance $\dist'_T(c,x)$ may be distorted badly. However, as discussed above, we only case about $x\in S$ which is ensured that the distance to $c$ is not that far. 
    
    Notice that for a fixed $c\in T$, the rounded distance function $\dist(c,\cdot)$ is determined by the values of $\dist_T'(c,x)$ for all $x\in I_T$, and each of them is among $\frac{6\lambda k}{\alpha\beta} + \frac{\lambda m}{\alpha}$ possible values. Therefore the number of distinct rounded distance functions $|\{\dist_T'(c,\cdot):c\in T\}|$ w.r.t. $T$ is 
    \begin{equation*}
        \le \left(\frac{2\lambda k}{\alpha\beta}+\frac{\lambda m}{\alpha} \right)^{|I_T|}\le \left(\frac{k}{\alpha\beta}+\frac{m}{\alpha}\right)^{O(\tw)}
    \end{equation*}
    We define $\dist':\calX\times X\to \R_+$ as: for every $c\in \calX$, let $T_c\in \mathcal{T}$ that contains $c$, it holds that $\dist'(c,x)=\dist'_{T_c}(c,x)$. Then we have 
    \begin{equation*}
        |\{\dist'(c,\cdot):c\in \calX\}|\le |\mathcal{T}_S|\cdot \left(\frac{k}{\alpha\beta}+\frac{m}{\alpha}\right)^{O(\tw)}\le \poly(|S|)\cdot \left(\frac{k}{\alpha\beta}+\frac{m}{\alpha}\right)^{O(\tw)}.
    \end{equation*}
    For every $C\in \C'(X)$, we denote by $\tilde v^C\in \R^X$ the \emph{rounded coset vector} defined as follows: for every $x\in X$, $\tilde v^C_x:=\min_{c\in C} \dist'(c,x)-\dist(x,C^\star)$. Let $\tilde{V}$ denote the set $\{\tilde v^C:C\in C'(X)\}$, we construct $U$ as follows: for every $\tilde v\in \tilde{V}$, $U$ contains one $v^C$ for some $C\in\C'(X)$ such that $\tilde v^C=\tilde v$.
    
    We next show that $U$ is the desired $\alpha$-covering.
    
    \paragraph{Cardinality of~$U$} Observe that the cardinality of $U$ is upper bounded by the number of distinct rounded distance functions $|\{\min_{c\in C}\dist'(c,\cdot): C\in \calX^k \}|$, which is at most 
    \begin{equation*}
        \left(\poly(|S|)\cdot \left(\frac{k}{\alpha\beta}+\frac{m}{\alpha}\right)\right)^{O(\tw)\cdot k} \le \left(\frac{k}{\alpha\beta}+\frac{m}{\alpha}\right)^{O(k\cdot \tw)}
    \end{equation*}
    Thus it holds that $\log |U|\le O(k\cdot \tw\cdot \log(\frac{k}{\alpha\beta}+\frac{m}{\alpha}) )$.
    
    \paragraph{Bouneded Covering Error}
    For every $C\in \C'(X)$, and $x\in S$, we have $\dist(x,C)\le \left(\frac{6k}{\beta}+m \right)\cdot\frac{\lambda}{n}\OPT(X)$. Hence, the error incurred by rounding, i.e., $|\dist'(c,x)-\dist(c,x)|$ is at most $\frac{\alpha}{n}\OPT(X)$. Let $\tilde v^C$ be the rounded cost vector with respect to $v^C$, by definition, there exists $C'\in \C'(X)$ such that $v^{C'}\in U$ and $\tilde v^{C'}=\tilde v^C$. The covering error is 
    \begin{align*}
        \left|v^C_x-v^{C'}_x\right|&\le \left|v^C_x-\tilde v^C+\tilde v^{C'}-v^{C'}_x\right|\\
        &\le \left|v^C_x-\tilde v^C \right| + \left|\tilde v^{C'}-v^{C'}_x \right|\\
        &\le \frac{2\alpha}{n}\OPT(X)\\
        &\le 2\alpha\cdot \err(x,v^C).
    \end{align*}
    It suffices to rescale $\alpha$.
    
    \paragraph{Bounded $L_\infty$ Norm} For every $v^C\in U$, we have $v^C\in V$ also, thus the $L_\infty$ norm is bounded because the $L_\infty$ norm of cost vectors in $V$ is bounded.
    \end{proof}
    \begin{proof}[Proof of \Cref{thm:various_metric} in graph metric space]
        Similar to doubling metric case, we plug the covering number in~\Cref{lem:covering_graph_metric} into \eqref{eq:simple_m_X} to concludes the proof.
    \end{proof}

\section{Lower Bounds}
\subsection{An $\Omega(1/\beta)$ Query Complexity Lower Bound for Any Algorithm}
\label{sec:proof_intro_lb}
\begin{theorem}[Restatement of \Cref{thm:intro_lb}]
    \label{thm:lb}
There exists a family of datasets $X \subset \mathbb{R}$ with balancedness $\beta$ such that any (randomized) $O(1)$-approximate algorithm for \twoMedian with success probability at least $3/4$ must query the identify of data points in $X$ for $\Omega(1/\beta)$ times (provided that queried points have free access to distance function). 
\end{theorem}
To prove our lower bound, we first apply the Yao's principle~\cite{yao1983lower} and derive the following lemma that reduce to proving lower bounds for deterministic algorithms with respect to some input distribution.
\begin{lemma}
    \label{lem:yao}
    For real number $\alpha>1$,
    let $D$ be a distribution over a family of datasets $X\subset \R$, 
    if any deterministic algorithm must query $\Omega(1/\beta)$ times to computes an $\alpha$-approximate center set for \twoMedian on $X$ sampled from $D$ with success probability at least $3/4$,
    then any randomized $\alpha$-approximate algorithm for \twoMedian with success probability at least $3/4$ must query $\Omega(1/\beta)$ times.
\end{lemma} 
\begin{proof}
    For an algorithm $\mathcal{A}$, we denote by $\mathcal{A}(X)$ the output of $A$ when running on $X$.
    For any randomized algorithm $\mathcal{R}$ which makes at most $o(1/\beta)$ queries, it can be seen as a distribution over some deterministic algorithms $\mathcal{R}_1,\dots,\mathcal{R}_s$. By assumption, for every $i\in [s]$, it holds that
    \begin{equation*}
        \Pr_{X\sim D}\left[\cost(X,\mathcal{R}_i(X)) > \alpha \OPT(X)\right]> 1/4.
    \end{equation*}
    Therefore, averaging over all deterministic algorithms, we have
    \begin{equation*}
        \Pr_{X\sim D}\left[\cost(X,\mathcal{R}(X)) > \alpha \OPT(X)\right]> 1/4.
    \end{equation*}
    where the randomness is from the choice of $X$ and the randomness of $\mathcal{R}$.
    Hence, there exists an instance $X\in\mathrm{supp}(D)$ such that 
    \begin{equation*}
        \Pr\left[\cost(X,\mathcal{R}(X)) > \alpha \OPT(X)\right]> 1/4.
    \end{equation*}
    which completes the proof.
\end{proof}
\begin{proof}[Proof of \Cref{thm:lb}]
    For an integer $n\ge 10/\beta$, let $m=\beta n/2$, we construct a family of $n$-point datasets as follows. For a set $\Pi\in [n]$ and a real number $t$, let $X^{\Pi,t}:=\{x_1^{\Pi,t},\dots,x_n^{\Pi,t}\}$ denotes an ordered $n$-point set in $1$-dimensional space such that,
    \begin{equation*}
        \forall i\in [n], \quad x_i^{\Pi,t}=\begin{cases}
            t& i\in \Pi\\
            0& i\not\in\Pi
        \end{cases}.
    \end{equation*}
    Let $\mathcal{X}:=\{X^{\Pi,t};|\Pi|=m, t\in \{1,2\}\}$ be a family of $n$-point datasets. Clearly, every $X\in\mathcal{X}$ is of balancedness $\beta$ for \twoMedian, and the optimal objective value of $X$ for \twoMedian is $0$. Let $D$ be a uniform distribution over $\mathcal{X}$.

    By \Cref{lem:yao}, it suffices to prove that for any $\alpha > 1$ and any deterministic algorithm $\mathcal{A}$, which makes at most $\frac{1}{2\beta}$ queries, will compute a $2$-point center set $C$ such that $\cost(X,C)>\alpha\cdot \OPT(X)$ for $X$ sampled from $D$ with probability at least $1/4$. We assume $\mathcal{A}$ queries with index $i_1,\dots, i_l$ with $l\le \frac{1}{2\beta}$, and receives $x_{i_1},\dots,x_{i_l}$ from the oracle. We observe that $\mathcal{A}$ finds a $2$-center set $C$ such that $\cost(X,C)>\alpha\cdot \OPT(X)$ if and only if $\mathcal{A}$ finds the exact optimal solution (otherwise, we say $\mathcal{A}$ fails). Therefore, we have
    \begin{equation}
    \label{eq:lb_goal}
    \begin{aligned}
        &\quad\Pr_{X\sim D}\left[\cost(X,\mathcal{A}(X))>\alpha\cdot\OPT(X) \right]\\
        &=\Pr_{X\sim D}[\mathcal{A}\text{ fails}]\\
        &=\Pr_{X\sim D}\left[\mathcal{A}\text{ fails} \mid \forall j\in[l], X_{i_j}=0 \right]\Pr_{X\sim D}\left[\forall j\in[l], x_{i_j}=0 \right]\\
        &\quad +\Pr_{X\sim D}\left[\mathcal{A}\text{ fails} \mid \exists j\in[l], X_{i_j}\neq 0 \right]\Pr_{X\sim D}\left[\exists j\in[l], x_{i_j}\neq 0 \right]\\
        &\geq \Pr_{X\sim D}\left[\mathcal{A}\text{ fails} \mid \forall j\in[l], X_{i_j}=0 \right]\Pr_{X\sim D}\left[\forall j\in[l], x_{i_j}=0 \right]
    \end{aligned}   
    \end{equation}
    Recall that $\mathcal{A}$ makes queries with index $i_{1},\dots,i_{l}$, and thus the output $\mathcal{A}(X)$ depends only on $x_{i_j}$ for all $j\in[l]$. For $X$ sampled uniformly from $\mathcal{X}$, its optimal center set is $\{0,1\}$ with probability $0.5$ and $\{0,2\}$ with probability $0.5$. Therefore, condition on $\forall j\in[l], X_{i_j}=0$, $\mathcal{A}$ will fail with probability at least $0.5$, i.e.,
    \begin{equation}
        \label{eq:lb_1}
        \Pr[\mathcal{A}\text{ fails}\mid \forall j\in[l], X_{i_j}=0]\ge \frac{1}{2}
    \end{equation}
    since it can not determine the center other than $0$. It remains to bound $\Pr_{X\sim D}[\forall j\in[l], X_{i_j}=0]$.
    \begin{equation}
    \label{eq:lb_2}
    \begin{aligned}
        \Pr_{X\sim D}[\forall j\in[l], X_{i_j}=0]&={n-l\choose m}\cdot {n\choose m}^{-1}\\
        &=\frac{(n-l)!(n-m)!}{(n-l-m)!n!}\\
        &=\frac{(n-l-m+1)(n-l-m+2)\cdots(n-l)}{(n-m+1)(n-m+2)\cdots n}\\
        &=\left(1-\frac{l}{n-m+1} \right)\left(1-\frac{l}{n-m+2} \right)\cdots \left(1-\frac{l}{n}\right)\\
        &\ge \left(1-\frac{l}{n-m+1} \right)^{m}\\
        &\ge 1-\frac{lm}{n-m+1}\\
        &\ge 1-\frac{1}{2\beta}\cdot \frac{\beta n}{2}\cdot \left(\left(1-\frac{\beta}{2}\right)n + 1 \right)^{-1}\\
        &\ge \frac{1}{2}
    \end{aligned}
    \end{equation}
    where the sixth derivation is due to Bernoulli's inequality. Therefore, plugging \eqref{eq:lb_1} and \eqref{eq:lb_2} into \eqref{eq:lb_goal}, we have 
    \begin{equation*}
        \Pr_{X\sim D}\left[\cost(X,\mathcal{A}(X))>\alpha\cdot\OPT(X) \right]\ge \frac{1}{4}
    \end{equation*}
    which completes the proof.
\end{proof}

\subsection{Vanilla \kMedian on Uniform Sample Can Incur Big Error}
\label{sec:proof_kmedian_lb}

\begin{lemma}
    \label{lemma:kbmedian_lb}
    There exists a family of $0.5$-balanced datasets $X_n\subset \R$ with $|X_n|=n$ for any integer $n\ge 1$, such that letting $S_n$ be a set of $o(n)$ uniform samples from $X_n$, with probability at least $1/4$, it holds that $\cost(X_n,C^\star_n)\ge 1.01\cdot\OPT(X_n)$, where $C^\star_n$ is an optimal center set for \ProblemName{$3$-Median} on $S_n$.
\end{lemma}

\begin{proof}
    Our proof is constructive. For $n\ge 10$, assume $S_n$ has $f(n)$ points sampled uniformly and independently from $X_n$, where $3\le f(n) < n/2$. To simplify the proof, we construct $X_n$ with size of $O(n)$ (instead of $n$) as follows. We place $n$ points at $0$ and $n$ points at $1$. Also, we place $n$ points at $w$, $\frac{n}{1.01\cdot f(n)}$ points at $w + f(n)$, and let $w\to \infty$. Clearly, the optimal $3$-point center set $C^\star_{n}$ is $\{0,1,w\}$, and the objective value is $n/1.01$. 

    Recall that $S_n$ is a set of $f(n)$ uniform samples from $X_n$, we have the probability of that $w+f(n) \in S_n$ is 
    \begin{align*}
        \Pr[w+f(n)\in S_n] &= 1- \left(1 - \frac{n/(1.01\cdot f(n))}{3n + n/(1.01\cdot f(n))} \right)^{f(n)}\\
        &\ge 1 - \exp\left(-\frac{n/1.01}{3n + n/(1.01\cdot f(n))} \right) \\
        &\ge 1 - \exp\left(-\frac{1}{3.03 + 1/3} \right)\\
        &\ge 0.25
    \end{align*}
    where the third derivation is due to $f(n)\ge 3$.
    Condition on $w+f(n)\in S_n$, we observe that $S_n$ contains at most $f(n)$ points at $0$, at most $f(n)$ points at $1$, at most $f(n)$ points at $w$ and at least $1$ point at $w+f(n)$. In this case, the optimal $3$-point center set $C'$ must contain $w+f(n)$. On $X_n$, the optimal center set that contains $w+f(n)$ is either $\{0,w,w+f(n)\}$ or $\{1,w,w+f(n)\}$, both of which have an objective value $n\ge 1.01\cdot\OPT(X_n)$. Therefore, $f(n)$ must be greater than $n/2$, which completes the proof.
\end{proof}

\section{Experiments}

\begin{figure*}[t]
    \centering
    \begin{subfigure}[b]{0.3\textwidth}
        \centering
        \includegraphics[width=\textwidth]{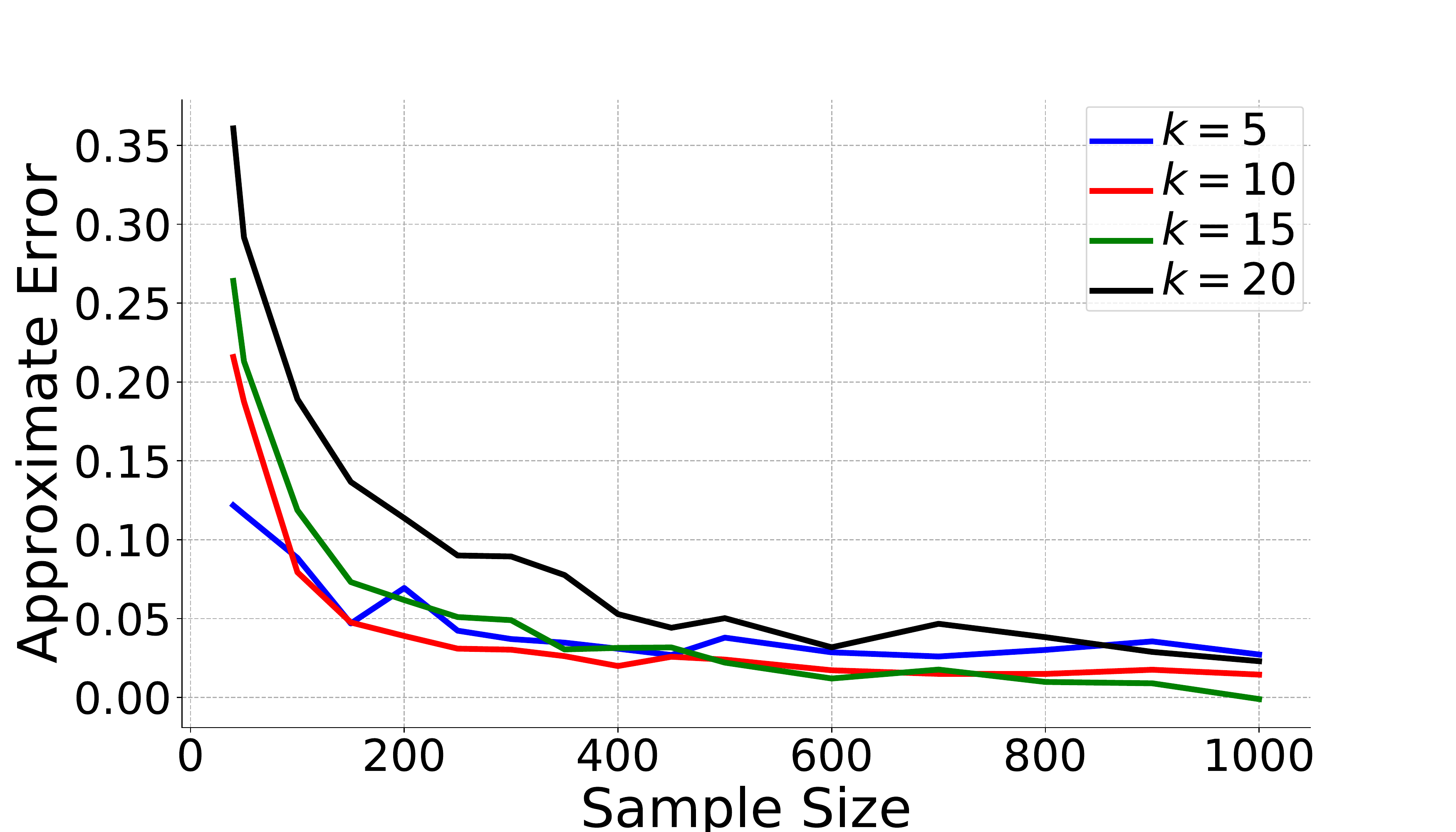}
        \caption*{Twitter dataset}
    \end{subfigure} \qquad
    \begin{subfigure}[b]{0.3\textwidth}
        \centering
        \includegraphics[width=\textwidth]{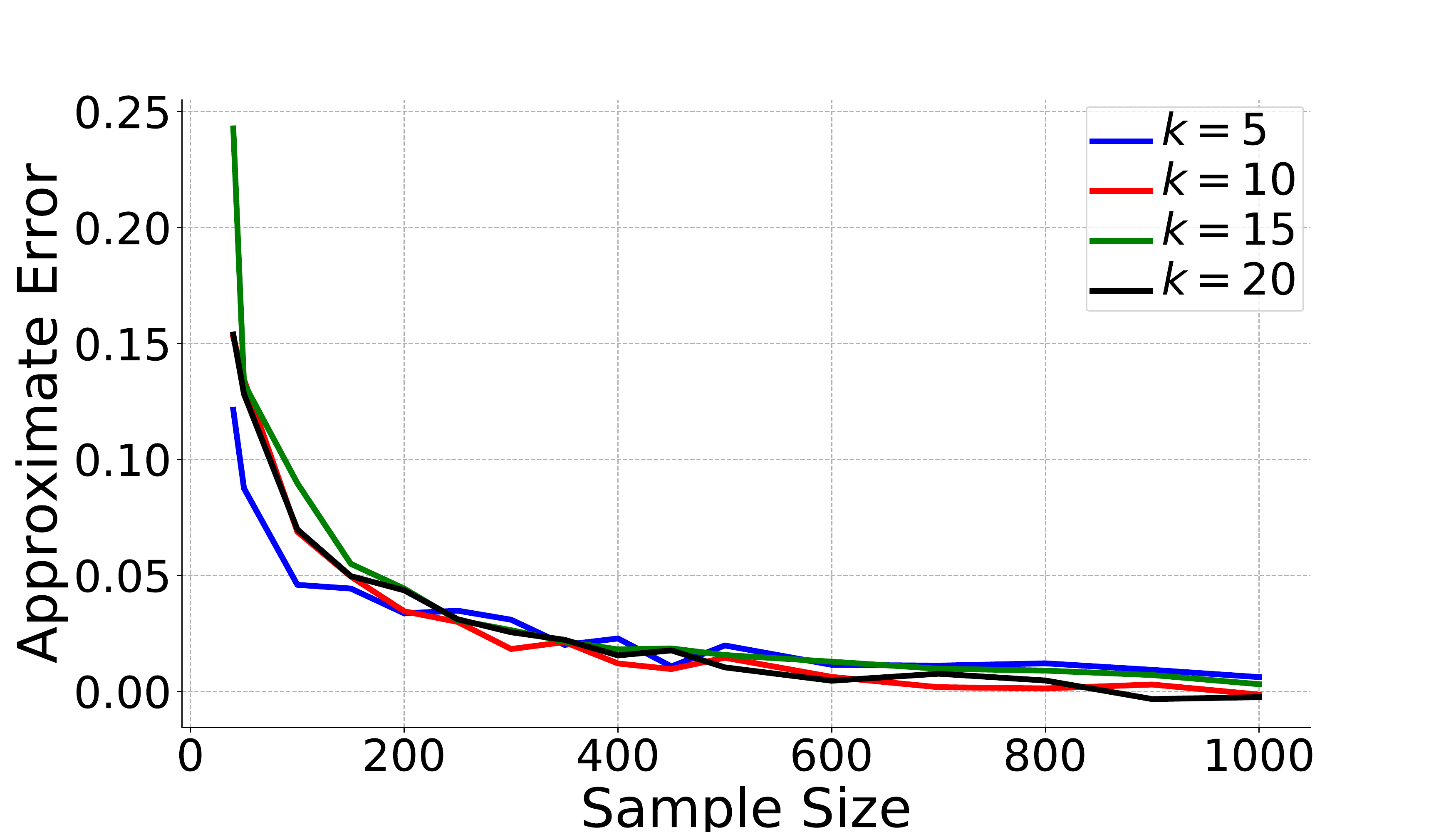}
        \caption*{Census1990 dataset}
    \end{subfigure}
    \begin{subfigure}[b]{0.3\textwidth}
        \centering
        \includegraphics[width=\textwidth]{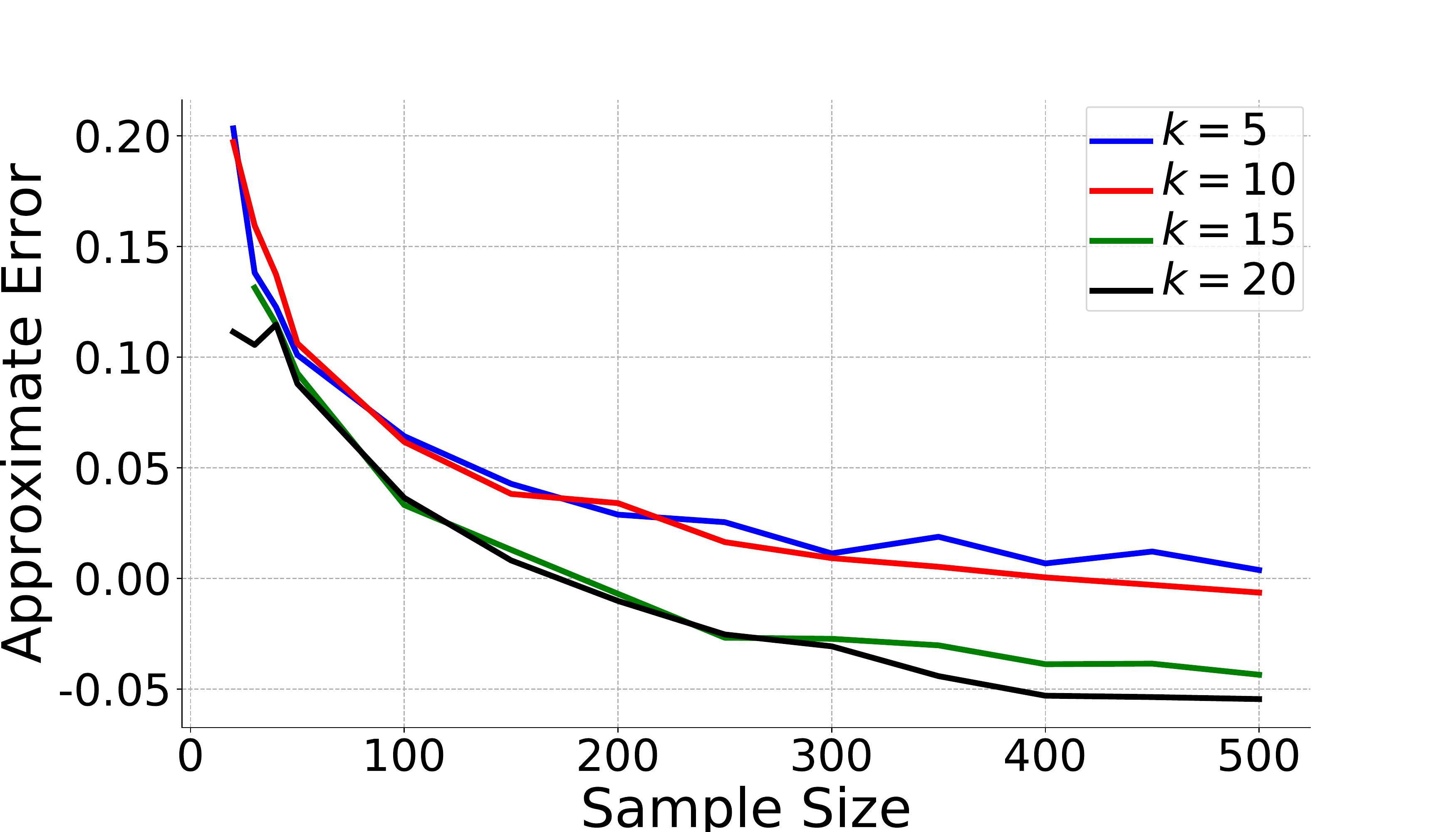}
        \caption*{NY dataset}
    \end{subfigure}
    \caption{The size-error tradeoff of the uniform sampling.}
    \label{fig:size_vs_error}
\end{figure*}

\begin{figure*}[t]
    \centering
    \begin{subfigure}[b]{0.3\textwidth}
        \centering
        \includegraphics[width=\textwidth]{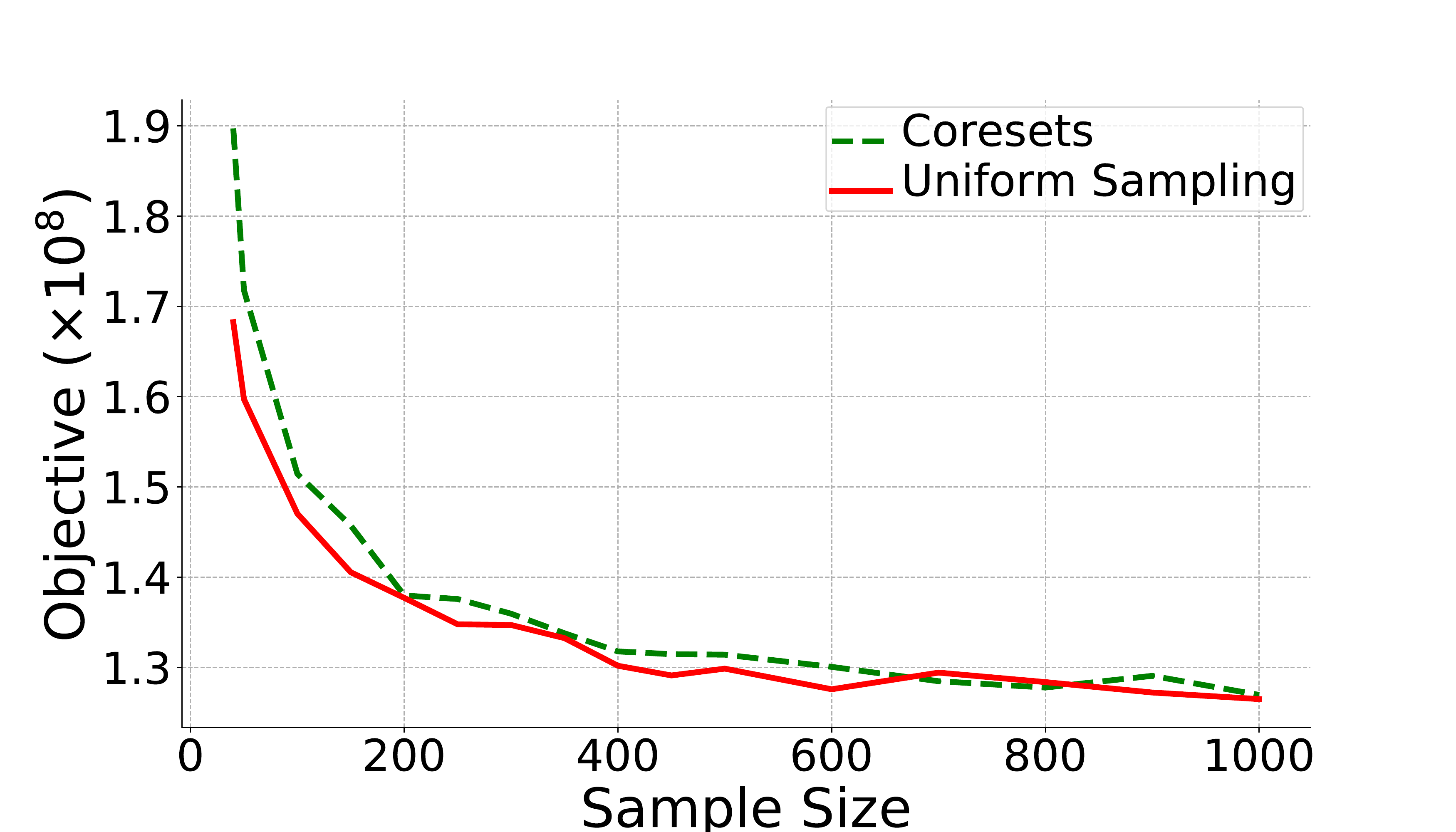}
        \caption*{Twitter dataset}
    \end{subfigure} \qquad
    \begin{subfigure}[b]{0.3\textwidth}
        \centering
        \includegraphics[width=\textwidth]{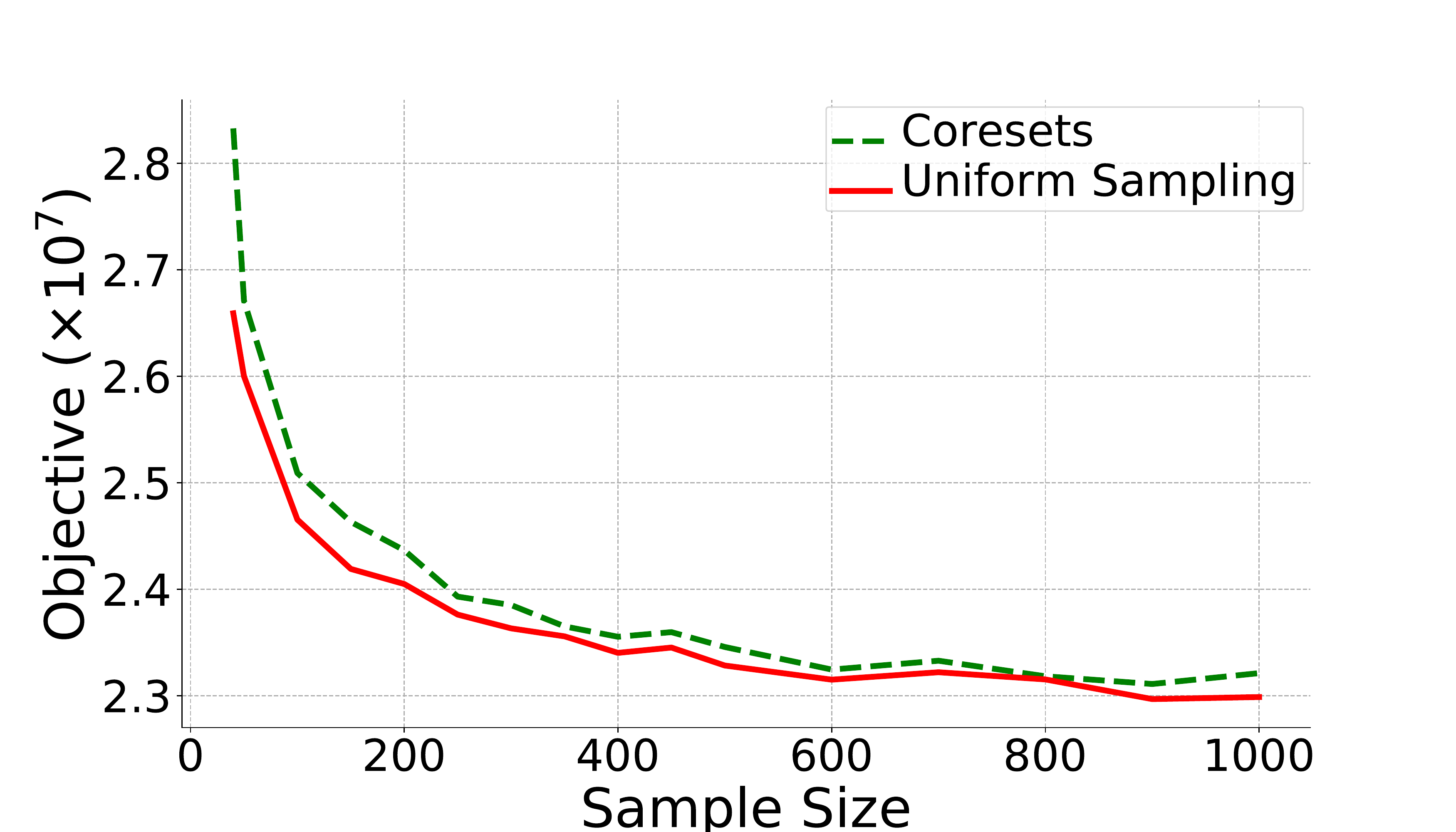}
        \caption*{Census1990 dataset}
    \end{subfigure}
    \begin{subfigure}[b]{0.3\textwidth}
        \centering
        \includegraphics[width=\textwidth]{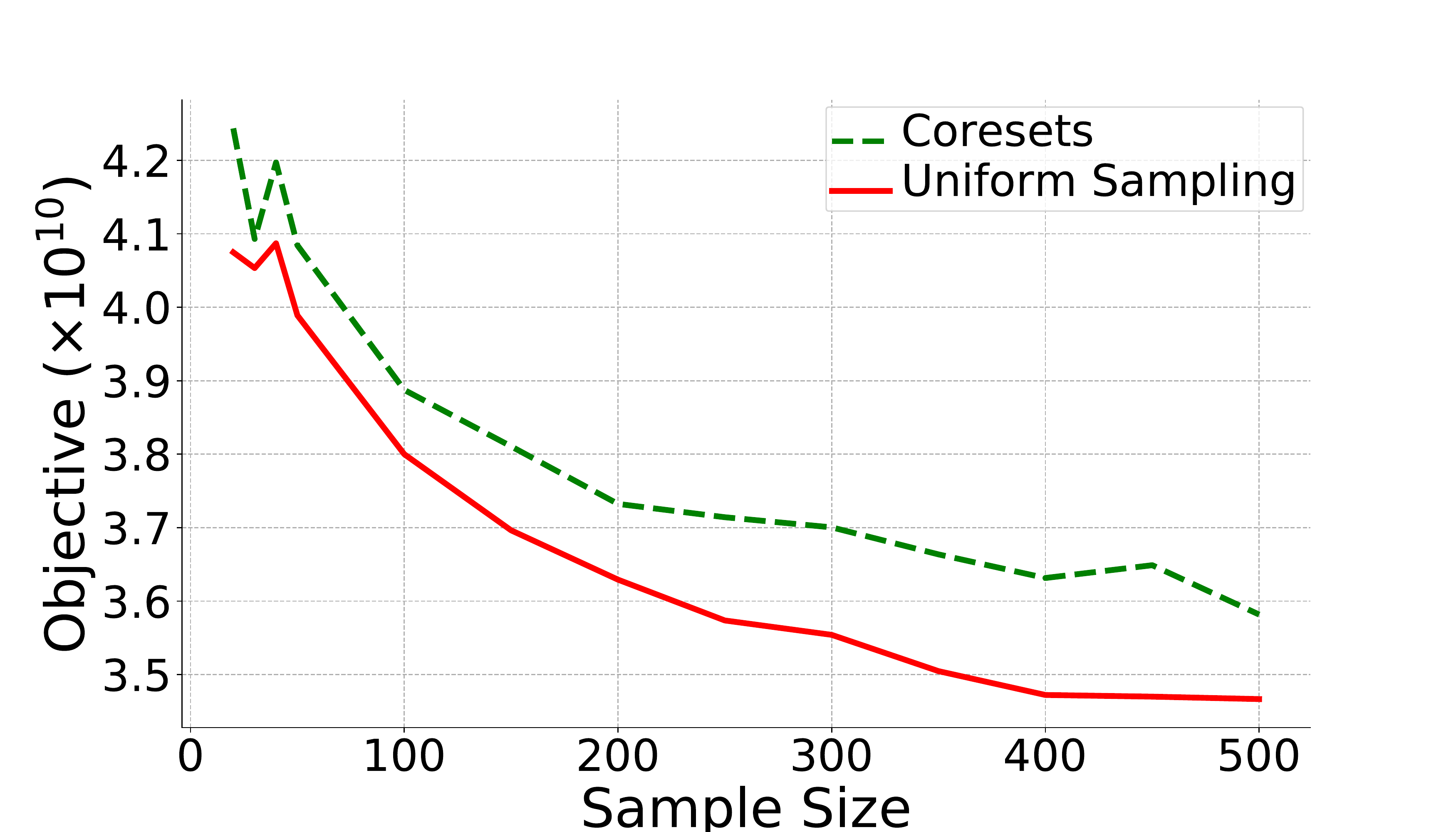}
        \caption*{NY dataset}
    \end{subfigure}
    \caption{The size-objective tradeoff, compared with that of coresets.}
    \label{fig:vs_coresets}
\end{figure*}

\paragraph{Experiment Setup}
Our experiments are conducted on $3$ real datasets: Twitter, Census1990 and NY.
For Twitter~\cite{twitter_data} and Census1990~\cite{census1990_data}, 
we select numerical features and normalize them into Euclidean vectors in $\mathbb{R}^d$, with $n=21040936$ data points and $d=2$ for Twitter, and $n=2458285, d=68$ for Census1990.
The NY dataset is a weighted graph with $|V| = 2276360$ vertices and $|E|=2974516$ edges,
representing the largest connected component of the road network of New York State, extracted from OpenStreetMap~\cite{OpenStreetMap}.
All experiments are conducted on a PC with Intel Core i7 CPU and 16 GB memory, and algorithms are implemented in C++ 11.

\paragraph{Error Measure}
To measure the error of a center set $C$ (which is a solution to \kMedian), ideally one should compare $\cost(X, C)$ with $\OPT$.
However, since the exact value of $\OPT$ is NP-hard to find, we turn to compare it with an approximate solution $\Capx$ instead.
Hence, our measure of error for a center set $C$ is defined as the following $\hat{\epsilon}$
which is a signed relative error compared with $\cost(X, \Capx)$: $\hat\epsilon(C) := \frac{\cost(X,C)-\cost(X,\Capx)}{\cost(X,\Capx)}$.
The reason to keep the sign is that $\Capx$ is not the optimal solution, hence it is possible that $\cost(X, C)$ is better than $\cost(X, \Capx)$ which leads to a negative sign.
To compute $\Capx$, we first construct a coreset (e.g., using~\cite{DBLP:conf/stoc/FeldmanL11}) to reduce the size of the dataset, and then apply a standard approximation algorithm for \kMedian~\cite{DBLP:conf/stoc/AryaGKMP01} on top of it. The use of coreset is to ensure the local search can run for sufficiently large number of rounds to converge in a reasonable time.

\paragraph{Experiment: Size-error Tradeoff and Balancedness}
Given $k \geq 1$, we take $m$ uniform samples $S$ from the dataset with varying $m$, and run a local search algorithm~\cite{DBLP:conf/stoc/AryaGKMP01} on $S$ to find a center set $C_S$ for \kMedian.
We evaluate the error $\hat\epsilon(C_S)$ and plot the size-error curves in \Cref{fig:size_vs_error} for various choices of $k$.
In fact, it turns out that the choice of $k$ also affect the balancedness parameter $\beta$.
Hence, for each dataset and value of $k$, we evaluate the balancedness $\beta$ of the abovementioned $\Capx$ on the original dataset,
and we also take $m = 500$ uniform samples $S$ and evaluate the balancedness $\beta'$ of $C_S$ on $S$.
The resulted $\beta$ and $\beta'$ are reported in \Cref{tab:balancedness}.
To make the measurement stable, we repeat the sampling and local search $20$ times independently and report the average statistics.

As can be seen from \Cref{fig:size_vs_error}, uniform sampling shows an  outstanding performance, and admits a similar convergence of the error curve regardless of the choice of $k$ and the consequent balancedness of datasets.
In NY dataset, it even achieves negative error when $k = 20$, which means it performs even better than the baseline.
We also observe in \Cref{tab:balancedness} that the datasets are mostly balanced even when $k$ is relatively large, and thus the factor of $1/\beta$ is reasonably bounded (i.e., no greater than $100$),
and even for the less balanced scenario, the error is still very small.
Moreover, the center set computed on $S$, although computed using a vanilla local search for \kMedian, actually satisfies the balancedness constraint of \kBMedian problem, that is, $\beta'\ge \beta/2$ for every choice of $k$ and every dataset.
This suggests that it is not necessary to run the more sophisticated \kBMedian at all in practice.
These findings help to justify why it is often seen in practice that a few uniform samples are enough to compute a good approximation.

\begin{table}[t]
    
    \caption{Balancedness under different values of $k$.}
    \label{tab:balancedness}
          \begin{center}
              \begin{small}
                  \begin{sc}
                      \begin{tabular}{ccccccc}
                          \toprule
                          \multirow{2}{*}{$k$} & \multicolumn{2}{c}{Twitter} & \multicolumn{2}{c}{Census1990} & \multicolumn{2}{c}{NY} \\
                          & $\beta$ & $\beta'$ & $\beta$ & $\beta'$ & $\beta$ & $\beta'$  \\
                          \midrule
                          5 & 0.549 & 0.523 & 0.302 & 0.292 & 0.036 & 0.127 \\
                          10 & 0.107 & 0.273 & 0.359 & 0.326 & 0.053 & 0.093\\
                          15 & 0.138 & 0.192 & 0.120 & 0.249 & 0.073 & 0.121 \\
                          20 & 0.202 & 0.184 & 0.045  & 0.216 & 0.072 & 0.120\\
                          \bottomrule
                      \end{tabular}
                  \end{sc}
              \end{small}
          \end{center}
      \vskip -0.1in
      \end{table}

\paragraph{Experiment: Comparison to Coreset}
Since uniform sampling is usually very efficient, we expect to see an advantage in running time compared with coreset constructions.
Our second experiment aims to demonstrate this advantage, particularly to compare with the coresets constructed by importance sampling~\cite{DBLP:conf/stoc/FeldmanL11}.
We set $k=20$, and vary the sample size $m$ for both uniform sampling and importance sampling.
Let $S$ and $S'$ be the subsets constructed by uniform sampling and importance sampling, respectively.
We run a local search algorithms on each of $S$ and $S'$, and compute the objective value of the output center set on the original dataset.
In \Cref{fig:vs_coresets}, we plot the objective-size curves, and observe that uniform sampling is comparable to coreset on all these datasets.
However, for the running time,
it takes $391$s/$114$s/$298$s to compute the coreset on Twitter/Census1990/NY dataset, while the runtime of the local step is $16$s/$42$s/$203$s.
Hence, the coreset construction is even more costly than the local search, and the coreset actually becomes the bottleneck of the total running time.
On the other hand, the uniform sampling takes merely $10^{-3}$s to sample $500$ points from the dataset which is significantly more efficient.
This also justifies the practical benefit of using uniform sampling than methods like importance sampling.

\bibliographystyle{alpha}
\bibliography{ref}

\end{document}